\documentclass[journal,twoside,web]{ieeecolor}
\usepackage{generic}
\usepackage{cite}
\usepackage{amsmath,amssymb,amsfonts}
\usepackage{algorithmic}
\usepackage{graphicx}
\usepackage{textcomp}

\usepackage{enumerate}
\usepackage{mathrsfs}
\usepackage{float}
\usepackage{cite}
\usepackage{tikz}
\usetikzlibrary{topaths,calc}
\usetikzlibrary{matrix,arrows,decorations.pathmorphing}
\usepackage{verbatim}
\usepackage{mathtools}
\usepackage{caption}
\usepackage{subcaption}
\usetikzlibrary{arrows}
\usepackage{tabularx}
\newtheorem{assumption}{Assumption}
\newtheorem{cor}{Corollary}
\newtheorem{exa}{Example}
\newtheorem{lemma}{Lemma}
\newtheorem{thm}{Theorem}
\newtheorem{prop}{Proposition}
\newtheorem{remark}{Remark}

\DeclareMathOperator*{\Dg}{\mathsf{diag}}
\DeclareMathOperator*{\argmax}{argmax}

\newcommand{\csx}[1]{\textcolor{blue}{#1}}
\newcommand{\hma}[1]{\textcolor{black}{#1}}

\def\BibTeX{{\rm B\kern-.05em{\sc i\kern-.025em b}\kern-.08em
    T\kern-.1667em\lower.7ex\hbox{E}\kern-.125emX}}
\begin{document}
\title{\LARGE \bf
On Metzler positive systems on hypergraphs }

\author{Shaoxuan Cui$^{1,4}$, Guofeng Zhang$^{2}$, Hildeberto Jardón-Kojakhmetov$^{1}$ and Ming Cao$^{3}$ 
\thanks{$^{1}$ S. Cui, and H. Jard\'on-Kojakhmetov are with the Bernoulli Institute for Mathematics, Computer Science and Artificial Intelligence, University of Groningen, Groningen, 9747 AG Netherlands {\tt\small \{s.cui, h.jardon.kojakhmetov\}@rug.nl}}
\thanks{$^{2}$ G. Zhang is with the Department of Applied Mathematics, The Hong Kong Polytechnic University, Kowloon 999077, Hong Kong, China and The Hong Kong Polytechnic University Shenzhen Research Institute, Shenzhen, Guang Dong 518057, China {\tt\small guofeng.zhang@polyu.edu.hk}}
\thanks{$^{3}$ M. Cao is with the Engineering and Technology Institute Groningen, University of Groningen, Groningen, 9747 AG Netherlands {\tt\small m.cao@rug.nl}}
\thanks{$^{4}$ S.Cui was supported by China Scholarship Council. The work of Cao was supported in part by the Netherlands Organization for Scientific Research (NWO-Vici-19902).}
}

\maketitle

\begin{abstract}
In graph-theoretical terms, an edge in a graph connects two vertices while a hyperedge of a hypergraph connects any more than one vertices. If the hypergraph's hyperedges further connect the same number of vertices, it is said to be \emph{uniform}. In algebraic graph theory, a graph can be characterized by an adjacency matrix, and similarly, a uniform hypergraph can be characterized by an adjacency \emph{tensor}.  This similarity enables us to extend existing tools of matrix analysis for studying dynamical systems evolving on graphs to the study of a class of polynomial dynamical systems evolving on hypergraphs utilizing the properties of tensors. To be more precise, in this paper, we first extend the concept of a Metzler matrix to a Metzler tensor and then describe some useful properties of such tensors. Next, we focus on positive systems on hypergraphs associated with Metzler tensors. More importantly, we design control laws to stabilize the origin of this class of Metzler positive systems on hypergraphs. In the end, we apply our findings to two classic dynamical systems: a higher-order Lotka-Volterra population dynamics system and a higher-order SIS epidemic dynamic process. The corresponding novel stability results are accompanied by ample numerical examples.
\end{abstract}

\begin{IEEEkeywords}
Hypergraphs, Higher-order interactions, Polynomial systems, H-eigenvalues, Perron-Frobenius Theorem, Stability, Feedback control
\end{IEEEkeywords}

\section{Introduction}
Polynomial dynamical systems \cite{craciun2019polynomial} are powerful tools to characterize a wide range of real-world complex systems, including biological, chemical, medical, engineering, and social systems. Chemical reactions \cite{donnell2013local,craciun2006understanding,angeli2009tutorial,ji2021autonomous} take place in energy, environmental, biological, and many other natural systems, and the inference of the reaction networks is of great significance for the understanding and design of chemical processes in engineering and life sciences \cite{ji2021autonomous}. In biology, the classic Lotka-Volterra model regarding the evolution of the species' population also belongs to the class of polynomial systems \cite{takeuchi1978stability,sb2010,goh1976global,goh1979stability}, where researchers treat the species \emph{pair} as a fundamental unit and only capture \emph{pairwise} interactions. The \emph{pairwise} interactions are associated with direct and additive effects of one species on another. In the classical model \cite{takeuchi1978stability,sb2010,goh1976global,goh1979stability}, one assumes interaction is always pairwise. Since the adjacency matrix of a graph usually represents the pairwise correspondence of species, the whole model can be abstracted as a system on a graph. Matrices and graphs are very useful tools for analysis of the Lotka-Volterra model \cite{sb2010}. Regarding epidemics and information diffusion in social networks, networked compartment models such as SIS\cite{liu2020stability,cui2022discrete,pare2018analysis,liu2019analysis}, SIR \cite{zhang2023analysis,zhang2023estimation}, SIRS \cite{liu2016node,liu2021optimal} are all polynomial systems. In the classical setting, the interaction among people (one sick person infects another person) is set to be pairwise as well. Thus, matrices and graphs also play an important role in the analysis of these systems \cite{liu2020stability,cui2022discrete,liu2019analysis,pare2018analysis,zhang2023analysis,zhang2023estimation,liu2016node,liu2021optimal}. 

Pairwise interactions and their representations as a graph may not always be appropriate because they exclude the possibility of group-wise interactions i.e. \emph{higher-order} interactions. In a recent study in ecology \cite{mayfield2017higher}, \emph{HOIs} (Higher-order Interactions) are introduced to represent non-additive effects among species and the existence of which is supported by empirical evidence, e.g. the one on natural plant communities \cite{mayfield2017higher}. The paper 
\cite{letten2019mechanistic} brings  HOIs into the Lotka–Volterra competition model and then demonstrates the effect of HOIs, by performing simulations with empirical data, showing that HOIs appear under almost all assumptions and help to improve the accuracy of predictions. In parallel, for compartment epidemic models, \cite{iacopini2019simplicial} argues that the social contagion processes, like opinion formation, and information diffusion, may depend more on group-wise interactions, which can be captured by HOIs. Following this idea, several novel epidemic models \cite{iacopini2019simplicial,cisneros2021multigroup,cui2023general,st2021universal,li2022competing} are developed by considering HOIs. While the pairwise interactions refer to a graph, the higher-order interactions usually correspond to the higher-order network \cite{bick2023higher} (sometimes referred to as simplicial complexes and hypergraphs). However, for the analysis of the higher-order Lotka-Volterra and epidemic models, there still remains an open question on how the hypergraph will influence the system evolving on it. Most current research \cite{iacopini2019simplicial,st2021universal,li2022competing,letten2019mechanistic} either relies on simulations or performs model reduction. Thus, a full analysis has not been carried out.

From a network perspective, a coupled cell system is a network of dynamical systems, where the fundamental units, ``cells'', are coupled together. Such systems are widely considered in control engineering \cite{lin2008distributed}, and mathematics \cite{bick2023higher,stewart2003symmetry,golubitsky2002patterns}. They can be represented schematically by a directed network whose nodes correspond to cells and whose edges represent couplings \cite{stewart2003symmetry}. A system composed of many agents interacting with each other can be naturally considered as a coupled cell system \cite{lin2008distributed}. A typical example of a coupled cell system with polynomial coupling is chemical reaction networks \cite{lin2008distributed}. Coupled cell
systems on hypergraphs are proposed in \cite{bick2023higher} by considering higher-order interactions and introducing the concept of hypergraphs. However, there is no general way to analyze such hypergraph systems.

It is known that the behavior of a system on a (conventional) graph is closely related to its corresponding graph, which is captured, for example, by an adjacency matrix. Among the matrices frequently used to model interesting real-life behavior, one of the most important matrices is the Metzler matrix \cite{FB-LNS}, whose off-diagonal elements are all non-negative. Metzler matrices are commonly associated with cooperative systems \cite{sb2010}, such as the cooperative Lotka-Volterra \cite{sb2010}, and single-virus systems \cite{liu2020stability,cui2022discrete,pare2018analysis}. For the analysis of such systems, the Perron–Frobenius theorem \cite{FB-LNS} is a fundamental tool, which indicates that an irreducible nonnegative matrix has a positive eigenvalue corresponding to a positive eigenvector. Analogously, it is shown that a hypergraph can be represented by an adjacency tensor \cite{bick2023higher,gallo1993directed}. Furthermore, the tensor version of a Perron–Frobenius theorem is also available \cite{chang2013survey,chang2008perron,yang2010further,yang2011further}, which serves as a potential suitable tool to study systems on hypergraphs. However, to the best of our knowledge, the Perron–Frobenius theorem of a tensor has never been used to study a dynamical system.

Besides the analysis of a networked system, one uses matrix theory to design control laws. In modern control theory, the $H_\infty$ controller \cite{gershon2005h} and LQR-controller \cite{okyere2019lqr} are two typical examples. Moreover, \cite{rantzer2011distributed} proposes a distributed control law for a class of linear positive systems with the help of the properties of a Metzler matrix. Then, a natural question arises, namely whether we can further utilize the properties of a tensor to design a control law for a class of systems on a hypergraph.

The contributions of this paper are briefly summarized as follows. Firstly, we introduce the concept of a Metzler tensor, develop the corresponding Perron–Frobenius theorem of an irreducible Metzler tensor and explain its relationship with $\mathcal{M}$-tensors studied in e.g.  \cite{zhang2014m,ding2013m,ding2016solving}. Secondly, we use the Perron–Frobenius theorem to construct a Lyapunov-like function to study a class of positive systems on hypergraphs. To illustrate our contributions, we apply our techniques to two real systems on hypergraphs, a higher-order Lotka-Volterra and a higher-order SIS model. Thirdly, we develop a feedback control strategy to stabilize the origin of a dynamical system on a hypergraph, and further study the case of a constant control input. Finally, all the theoretical results are supported by numerical examples.

\emph{Notation:} $\mathbb{C}$ and $\mathbb{R}$ denote the set of complex and real numbers, respectively. $\mathbb{R}_+$ denotes the set of positive real numbers. For a matrix $M \in \mathbb{R}^{n \times r}$ and a vector $a \in \mathbb{R}^n$, $M_{ij}$ and $a_{i}$ denote the element in the $i$th row and $j$th column and the $i$th entry, respectively. Given a square matrix $M \in \mathbb{R}^{n \times n}$, $\rho(M)$ denotes the spectral radius of $M$, which is the largest absolute value of the eigenvalues of $M$. The notation $|M|$ denotes the matrix whose entry $|M|_{ij}$ is the absolute value of $M_{ij}$. For any two vectors $a, b \in \mathbb{R}^n$, $a > (<) b$ represents that $a_i >(<) b_i$, for all $i=1,\ldots,n$; and $a \geq (\leq) b$ means that $a_i \geq (\leq) b_i$, for all $i=1,\ldots,n$. These component-wise comparisons are also used for matrices or tensors with the same dimension. The vector $\mathbf{1}$ ($\mathbf{0}$) represents the column vector or matrix of all ones (zeros) with appropriate dimensions. The previous notations have straightforward extensions to tensors. In the following section, we introduce the preliminaries on tensors and some further notations regarding tensors. \csx{Furthermore, the notations frequently used in this paper are listed in the table \ref{tab:notation}.}

\begin{table}[]
\caption{Notations frequently used in this paper}
\label{tab:notation}
\begin{tabular}{@{}ll@{}}
\hline
$\mathbf{1}$ &The all-one tensor (matrix, vector) with appropriate dimensions   \\ 
$\mathbb{R}$ & Set of real numbers  \\
$\mathbf{0}$ & The all-zero tensor (matrix, vector) with appropriate dimensions\\
$|M|$& The tensor (matrix) whose entry is the absolute value of \\
&the entry of $M$.\\
$\rho(M)$ & The spectral radius of a cubical tensor (square matrix) $M$ \\ 
$\mathbb{R}^{[n,k]}$ & $n$ is the dimension of a tensor and $k$ is the order.\\
$A x^{{k}-1}$ & Tensor-vector product to the power of $k-1$.\\
$x^{[{k}-1]}$ & Vector-vector product to the power of $k-1$.\\
\hline
\end{tabular}
\end{table}

\section{Preliminaries on tensors and hypergraphs}

A tensor $T\in\mathbb{R}^{n_1\times n_2 \times \cdots \times n_k}$ is a multidimensional array, where the order is the number of its dimension $k$ and each dimension $n_i$, $i=1,\cdots,k,$ is a mode of the tensor. A tensor is cubical if every mode has the same size, that is $T\in\mathbb{R}^{n\times n \times \cdots \times n}$. We further write a $k$-th order $n$ dimensions cubical tensor as $T\in\mathbb{R}^{n\times n \times \cdots \times n}=\mathbb{R}^{[n,k]}$. A cubical tensor $T$ is called supersymmetric if $T_{j_1 j_2 \ldots j_k}$ is invariant under any permutation of the indices. For the rest of the paper, a tensor always refers to a cubical tensor unless specified otherwise.

The identity tensor $\mathcal{I}=\left(\delta_{i_1 \cdots i_m}\right)$ is defined as
\begin{equation*}
   \delta_{i_1 \cdots i_m}= \begin{cases}1 & \text { if } i_1=i_2=\cdots=i_m \\ 0 & \text { otherwise }\end{cases} .
\end{equation*}
The diagonal entries of a tensor are the entries with all the same index, for example, $A_{iiiiii}$. All other entries are called off-diagonal entries.

We then consider the following notation: the tensor-vector product to the power $k-1$, where $k$ is the order of the tensor, $A x^{k-1}$
  and the vector-vector product to the power $k-1$: $x^{[k-1]}$ are vectors, whose $i$-th components are
$$
\begin{aligned}
\left(A x^{{k}-1}\right)_i & =\sum_{i_2, \ldots, i_{{k}}=1}^n A_{i, i_2 \cdots i_{{k}}} x_{i_2} \cdots x_{i_{{k}}}, \\
\left(x^{[{k}-1]}\right)_i & =x_i^{{k}-1}.
\end{aligned}
$$
From the definition, we see that the product $Ax^{k-1}$ can capture an arbitrary homogeneous polynomial vector field since each entry of the product is a homogeneous polynomial. For example,
\begin{equation}\label{eq:example}
    \begin{split}
        \dot{x}_1&=-x_1^2+x_1 x_2,\\ \dot{x}_2&=x_1^2+x_1 x_2-x_2^2
    \end{split}
\end{equation}
can be written as $\dot{x}=Ax^{k-1}$, where $A_{111}=-1, A_{112}=1, A_{211}=1, A_{212}=1, A_{222}=-1$ and all other entries of $A$ are zero. In general, the choice of the tensor $A$ is not unique. If we choose $A_{111}=-1, A_{112}=0.5, A_{121}=0.5, A_{211}=1, A_{212}=1, A_{222}=-1$ and all other entries of $A$ are zero, then we get the same vector field \eqref{eq:example}.

Recall that $A x^{{k}-1} \text{ and }  x^{[{k}-1]}$ are both vectors. For a tensor, consider the following homogeneous polynomial equation:
\begin{equation}\label{eq:eigenproblem}
A x^{{k}-1}=\lambda x^{[{k}-1]},
\end{equation}
where if there is a real number $\lambda$ and a nonzero real vector $x$ that are solutions of \eqref{eq:eigenproblem}, then $\lambda$ is called an H-eigenvalue of $A$ and $x$ is the H-eigenvector of $A$ associated with $\lambda$ \cite{qi2005eigenvalues,zhang2014m,chang2013survey}. {The number and distribution of H-eigenvalues of a real supersymmetric tensor are studied by \cite{qi2005eigenvalues}. It will be briefly introduced in the later section \ref{sec:heig}.}
Throughout this paper, the words eigenvalue and eigenvector as well as H-eigenvalue and H-eigenvector are used interchangeably. It is also worth mentioning that there are still some other kinds of definitions of the eigenvalues and eigenvectors of a tensor, e.g. Z-eigenvalues \cite{chang2013survey,qi2005eigenvalues} and U-eigenvalues \cite{zhang2020iterative}. However, in this paper, we will always refer to H-eigenvalues and H-eigenvectors as defined above.

It is straightforward to check that $A x^{{k}-1}$ and $x^{[{k}-1]}$ satisfy the following
: if $Ax^{k-1}=k_1 x^{[k-1]}$ and $Bx^{k-1}=k_2 x^{[k-1]}$, then $Ax^{k-1}+Bx^{k-1}=(A+B)x^{k-1}=k_1 x^{[k-1]}+k_2 x^{[k-1]}=(k_1+k_2)x^{[k-1]}$.

The spectral radius of the tensor $A$ is defined as
$
\rho(A)=\max \{|\lambda|: \lambda \text { is an eigenvalue of } A\}.
$
A tensor $\mathcal{C}=\left(c_{{i_1} \ldots {i_m}}\right)$ of order $m$ dimension $n$ is called reducible if there exists a nonempty proper index subset $I \subset\{1, \ldots, n\}$ such that
$$
c_{i_1 \cdots i_m}=0 \quad \forall i_1 \in I, \quad \forall i_2, \ldots, i_m \notin I .
$$
If $\mathcal{C}$ is not reducible, then we call $\mathcal{C}$ irreducible. A tensor with all non-negative entries is called a non-negative tensor. Let $A$ be an $m$-order and $n$-dimensional tensor. The tensor $A$ is called an $\mathcal M$-tensor if there exist a nonnegative tensor $B$ and a positive real number $\eta \geq \rho(B)$ such that
$$
A=\eta \mathcal{I}-B
$$
If $\eta>\rho(B)$, then $A$ is called a non-singular $\mathcal M$-tensor. {Moreover, a tensor ${A}$ is called diagonally dominant if
$$
\left|A_{i i \ldots i}\right| \geq \sum_{\left(i_2, \ldots, i_m\right) \neq(i, \ldots, i)}\left|A_{i i_2 \ldots i_m}\right| \quad \text { for all } i=1,2, \ldots, n;
$$
and is called strictly diagonally dominant if
$$
\left|A_{i i \ldots i}\right|>\sum_{\left(i_2, \ldots, i_m\right) \neq(i, \ldots, i)}\left|A_{i i_2 \ldots i_m}\right| \quad \text { for all } i=1,2, \ldots, n.
$$

We now recall the Perron–Frobenius Theorem for irreducible nonnegative tensors:
\begin{lemma}[Theorem 3.6 \cite{chang2013survey}]\label{lem:perron}
If $A$ is a nonnegative irreducible tensor of order $m$ dimension $n$ and $\rho(A)$ is its spectral radius, then the following hold:
\begin{itemize}
    \item [(1)] $\rho(A)>0$.
\item[(2)] There is a strictly positive eigenvector $x >\mathbf{0}$ corresponding to $\rho(A)$.
\item[(3)] If $\lambda$ is an eigenvalue with nonnegative eigenvector, then $\lambda=\rho(A)$. Moreover, the nonnegative eigenvector is unique up to a multiplicative positive constant.
\item[(4)] If $\lambda$ is an eigenvalue of $A$, then $|\lambda| \leq \rho(A)$.
\end{itemize}
\end{lemma}\bigskip

\begin{remark}
    There is also a weak version of the Perron–Frobenius Theorem (Theorem 3.5 \cite{chang2013survey}), which is applicable to nonnegative tensors with a nonnegative eigenvector. Since we shall concentrate on positive systems, the version in Lemma \ref{lem:perron} is more useful for our purposes. Note that Lemma \ref{lem:perron} is only valid for tensors' H-eigenvalues.
\end{remark}

Furthermore, for nonnegative tensors, we have the following properties.

\begin{lemma}[Theorems 2.19 and 2.20 \cite{yang2011further}]\label{lem:2}
   Let $A,\,B$ be nonnegative tensors. If $0 \leq B \leq A$, then $\rho(B) \leq \rho(A)$. If $0 \leq B \leq A$, where $A$ is irreducible and $A \neq B$, then $\rho(A)>\rho(B)$.
\end{lemma}


In the next, we illustrate some definitions regarding hypergraphs.

The concept of a hypergraph is defined in \cite{gallo1993directed}. Here, we utilize a set of tensors to represent the information of the weights of a hypergraph. A weighted and directed hypergraph is a triplet\csx{ $\mathscr{H}=(\mathcal{V},\mathcal{E}, \Tilde{A})$}. The set $\mathcal{V}$ denotes a set of vertices and $\mathcal{E}=\{E_1, E_2, \cdots,E_n \}$ is the set of hyperedges. A hyperedge is an ordered pair $E=(\mathcal{X},\mathcal{Y})$ of disjoint subsets of vertices; $\mathcal{X}$ is the tail of $E$ and $\mathcal{Y}$ is the head.
As a special case, a weighted and undirected hypergraph is a triplet $\mathscr{H}=(\mathcal{V},\mathcal{E}, \Tilde{A})$, where $\mathcal{E}$ is now a finite collection of non-empty subsets of $\mathcal{V}$. If all hyperedges of the hypergraph contain the same number of (tails, heads) nodes, then the hypergraph is uniform. For details, see \cite{chen2021controllability} for undirected uniform hypergraphs and see \cite{xie2016spectral} for directed uniform hypergraphs. {The connection between undirected hyperedges and directed hyperedges and their physical meanings are introduced in \cite{gallo2022synchronization}. Briefly speaking, for example, an undirected hyperedge of 3 nodes $i,j,k$ which can capture the overall interaction effect when $i,j,k$ are interacting can be decomposed by three kinds of directed hyperedges, namely, $j,k$s' influence on $i$ ($A_{ijk},A_{ikj}$), $i,j$s' influence on $k$ ($A_{kij},A_{kji}$), $i,k$s' influence on $j$ ($A_{jik},A_{jki}$).} For modeling purposes, regarding the nodal dynamics, one directed hyperedge usually denotes the joint influence of a group of agents on one agent. Thus, it suffices to deal with hyperedges with one single tail, and we assume that each hyperedge has only one tail but one or multiple ($\geq 1$) heads. This setting is consistent with \cite{xie2016spectral} and has the advantage that a directed uniform hypergraph can be represented by an adjacency tensor.  In section \ref{sec:modeling}, we will show why this setting is sufficient for modeling nodal dynamics (evolution of nodes). \csx{For now, we mention that from a modeling perspective, the hyperedges with multiple tails may be helpful when one wants to capture the influence of one group of agents on another one. If the interactions of the model are supersymmetric, then an undirected hypergraph can be adopted instead of a directed one. For example, the influence of $i,j$ on $k$ is the same as $j,k$ on $i$ and $i,k$ on $j$.} Generally, an undirected uniform hypergraph can be represented by a supersymmetric adjacency tensor. 
We now introduce the set of tensors ${\Tilde{A}}=\{A_2, A_3,\cdots\}$ to collect the weights of all hyperedges, where $A_2=[A_{ij}]$ denotes the weights of all second-order hyperedges, $A_3=[A_{ijk}]$ denotes the weights of all third-order hyperedges, and so on. For instance, $A_{ijkl}$ denotes the weight of the hyperedge where $i$ is the tail and $j,k,l$ are the heads. For simplicity, in this paper, we also use the weight (for example, $A_{\bullet}$) to denote the corresponding hyperedge. If all hyperedges only have one tail and one head, then the network is a standard directed and weighted graph.

For convenience, throughout this paper, we define a multi-index notation $I=i_2,\cdots,i_k$, where $k$ is the order of the associated tensor.

\section{Perron– Frobenius Theorem of a Metzler Tensor}

Similar to the definition of a Metzler matrix \cite{narendra2010hurwitz}, we define a Metzler tensor as a tensor whose off-diagonal entries are non-negative. Then, it is straightforward to see that any Metzler tensor $A$ can be represented as $A=B-s\mathcal{I}$, where $s$ is a real scalar and $B$ is a nonnegative tensor.

\begin{prop}\label{prop:Metzler}
A Metzler tensor $A=B-s\mathcal{I}$ always has an eigenvalue that is real and equal to $\lambda=\lambda(B)-s$, where $\lambda(B)$ is an eigenvalue of a nonnegative tensor $B$.
\end{prop}



\begin{proof}
We remind the readers of the multi-index notation $I=i_2,\cdots,i_k$.
For a nonnegative tensor $B$, we have 
$$
\begin{aligned}
B x^{k-1}&=\lambda(B) x^{[k-1]},\\
\left(B x^{k-1}\right)_i & =\sum_{i_2, \ldots, i_k=1}^n B_{i, I} x_{i_2} \cdots x_{i_k}, \\
\lambda(B)\left(x^{[k-1]}\right)_i & =\lambda(B) x_i^{k-1}.
\end{aligned}
$$
Thus, it follows that 
{$$
\begin{aligned}
Ax^{k-1}&=(B-s\mathcal{I}) x^{k-1}=(\lambda(B)-s) x^{[k-1]},\\
\end{aligned}
$$}

Finally, we see indeed that $\lambda=\lambda(B)-s$.
\end{proof}

Next, we show the Perron-Frobenius theorem for an irreducible Metzler tensor. 

We have the following results.

\begin{thm}\label{thm:perron}
If $B$ is an irreducible nonnegative tensor of order $m$ dimension $n$ with $\rho(B)$ its spectral radius, then $A=B-s\mathcal{I}$ is an irreducible Metzler tensor of order $m$ dimension $n$. Moreover, $A$ has a particular eigenvalue $\lambda(A)$ given by $\lambda(A)=\rho(B)-s$ and called \emph{the Perron-H-eigenvalue of $A$}. Furthermore the following hold:
\begin{itemize}
\item[(1)] There is a strictly positive eigenvector $x >\mathbf{0}$ corresponding to $\lambda(A)$.
\item[(2)] If $\lambda$ is an eigenvalue of $A$ with nonnegative eigenvector, then $\lambda=\lambda(A)$. Moreover, such a nonnegative eigenvector is unique up to a positive multiplicative constant.
\end{itemize}
\end{thm}

\begin{proof}
The proof is straightforward. Notice that $\lambda(A)=\rho(B)-s$ and apply Lemma \ref{lem:perron}.
\end{proof}

\begin{remark}
    Throughout the rest of this paper, $\lambda(A)$ shall always refer to the Perron-H-eigenvalue of the tensor $A$. {If and only if $A$ is a Metzler tensor and $\lambda(A)\leq 0$ ($\lambda(A)< 0$), then $-A$ is an (non-singular) $\mathcal{M}$-tensor. Hence, Metzler tensors and $\mathcal{M}$-tensors are very similar concepts.}
    The properties of  $\mathcal{M}$-tensors are studied by\cite{zhang2014m,ding2013m,ding2016solving}. Metzler tensors and $\mathcal{M}$-tensors are very similar concepts. In the following, considering the connection between a Metzler tensor and a $\mathcal{M}$-tensor, we further utilize the properties of $\mathcal{M}$-tensors to achieve some of the main results. Speaking of the origin, as for matrices, Metzler matrices are frequently used in the context of system and control, while $\mathcal{M}$-matrices or tensors arise frequently in scientific computations. Since the concept of a Metzler tensor is missing, it is still meaningful to give a formal definition and discover the properties of a Metzler tensor in our paper. {There is also a similar concept of $\mathcal{Z}$-tensor \cite{zhang2014m}
    : $A$ is a $\mathcal{Z}$-tensor if and only if $-A$ is a Metzler tensor. To the best of our knowledge, none of these concepts have been previously exploited to analyze dynamical systems as we propose in this manuscript.}
\end{remark}

{\section{Polynomial dynamical systems on a hypergraph}\label{sec:modeling}}
{In this section, we give a brief introduction to the modeling framework of nodal dynamics on a hypergraph. Proposed by \cite{bick2023higher}, coupled cell
systems \cite{stewart2003symmetry} on a hypergraph are in the form of 
\begin{equation}\label{eq:cell}
\begin{split}
    \dot{x}_i&=F\left(x_i\right)+\sum_{j=1}^N (A_2)_{ij} G_i\left(x_i, x_j\right)\\
    &+\sum_{j, l=1}^N (A_3)_{ij l} G_i^{(3)}\left(x_i, x_j, x_l\right)+\cdots,
\end{split}
\end{equation}
where $F$ is a function that describes the intrinsic coupling of the node, the adjacency matrix $A_2$ together with the coupling functions $G_i$ represent the pairwise network interactions, and the coefficients of $A_s$ (which are adjacency tensors) and coupling function $G_i^{(s)}$ (with $s \geq 3$ ) are called higher-order network interactions. For example, $(A_3)_{i j l}$ and $G_i^{(3)}\left(x_k, x_j, x_i\right)$ describe the joint influence of nodes $l, j$ on node $i$, which can be captured by a directed hyperedge with a single tail. Since we focus on nodal dynamics, using a hypergraph with all edges with one single tail is sufficient. Then, we further assume that the intrinsic coupling $F$ can be represented by some self-arc $(A_s)_{i\cdots i}x_i^{s-1}$ and the coupling function $G_k^{(s)}(x_i, x_j, x_l,\cdots)=x_i x_j x_l\cdots$. Note that this multiplicative interaction has some important physical interpretations. For example, the probability of some independent event occurring at the same time is captured by the multiplicative interaction.
Then, it is straightforward to see that \eqref{eq:cell} can be written in a tensor form $\dot{x}=A_{k} x^{k-1}+A_{k-1} x^{k-2}+\cdots +A_{2} x$. It is worthwhile mentioning that any polynomial system, where the powers are positive integers, can be written in this form. Especially, any homogeneous polynomial system is in the form of $\dot{x}=A x^{k-1}$, which is described by a uniform hypergraph.}

{We further emphasize that polynomial coupling functions are widely used in many real-world systems on hypergraphs. For example, an epidemic model on a hypergraph and a generalized Lotka-Volterra model adopt polynomial coupling functions. We will introduce both systems later in sections \ref{sec:sis} and \ref{sec:lv} in detail. We further observe that system \eqref{eq:cell} is based on a network with a coupled cell structure. That is a network that describes how nodes are influenced and interact with each other. Sometimes, this hypergraph has a direct physical meaning, such as social networks (epidemic model) \cite{cui2023general} and ecological networks (Lotka-Volterra) \cite{letten2019mechanistic,cui2023species}. In addition, notice that a system may be represented by different networks. For example, for a chemical reaction system with $X+Y\rightarrow Z+R$, the evolution of $X$ is described by $\Dot{x}_1=-k_1 x_1 x_2$. According to the chemical reaction $X+Y\rightarrow Z+R$, it is straightforward to define a hyperedge where $X,Y$ are heads and $Z,R$ are tails. However, if one uses the concept of a coupled cell network, the evolution of $X$ is only related to itself and $Y$. Thus, we just need one hyperedge with a single tail $X$ and heads $X,Y$. Since the system dynamics are just the same, these are two different representations of the given model.
}

\csx{In the section \ref{sec:uni},} we focus on the particular case where the tensor has a Metzler structure and will utilize Theorem \ref{thm:perron} to analyze the stability of some positive systems on a uniform hypergraph.

\section{Hypergraph spectrum, strong connectivity, and H-eigenvector centrality}

\csx{In this section, we introduce the concept of hypergraph spectrum, strong connectivity, and H-eigenvector centrality. We know that a uniform hypergraph with hyperedges of $k$ nodes can be captured by an adjacency tensor of order $k$. In the literature \cite{benson2019three,cooper2012spectra}, all these concepts (spectrum, centrality) are for a signless hypergraph (all weights of its edges are non-negative). To the best of our knowledge, there is no literature about these concepts on a signed hypergraph with both positive and negative weights. In this paper, we deal with a tensor or a hypergraph with a Metzler structure. We know that $A=B-s\mathcal{I}$, where $s\in \mathbb{R}$, $A$ is a Metzler tensor and $B$ is non-negative. Instead of investigating $\mathscr{H}(A)$, we can study a signless hypergraph $\mathscr{H}(B)$. Later, we will show that this manipulation has very little influence on the properties of a hypergraph, i.e., $\mathscr{H}(A)$ and $\mathscr{H}(B)$ have some similar properties, where $\mathscr{H}(A)$ denotes a hypergraph with the adjacency tensor $A$.}

\csx{First of all, throughout the section, to make the choice of $B,s$ unique given a Metzler tensor $A$, we choose the smallest $s$ such that $B$ is non-negative. If $A$ is already non-negative, then $A=B, s=0$.}

\csx{In \cite{cooper2012spectra}, the spectrum of a signless uniform hypergraph $\mathscr{H}(B)$ is defined as the spectrum (spectral radius) of the corresponding adjacency tensor $\rho(B)$. According to Theorem \ref{thm:perron}, the Perron-eigenvalue of $A$ is a shift on the spectrum of $B, \mathscr{H}(B)$, i.e., $\lambda(A)=\rho(B)-s$. In fact, in the next section, we will show that the stability of a Metzler dynamical system on a uniform hypergraph is indeed related to $\rho(B)$ and $s$.}

\csx{Next, a $k$-uniform, $n$-node hypergraph with adjacency tensor ${{T}}$ is strongly connected if the graph induced by the $n \times n$ matrix ${M}$ obtained by summing the modes of ${{T}}, M_{i j}=\sum_{j_3 \ldots \ldots, j_k} {T}_{i, j, j_3, \ldots, j_k}$, is strongly connected. As a consequence, a strongly connected hypergraph has an irreducible adjacency tensor \cite{benson2019three,qi2017tensor}.}

\csx{Eigenvector centrality is a standard network analysis tool for determining the importance of entities in a strongly connected system that is represented by a network. In \cite{benson2019three}, an H-eigenvector centrality (HEC) of a uniform signless hypergraph is proposed. Notice that although the main results in \cite{benson2019three} are based on the setting of an unweighted hypergraph, they further claim that the unweighted setting is not necessary and that all of the proposed methods work seamlessly if the hypergraph
is weighted. In the next, we give further details regarding the weighted case. Let us start with an example of a $3$-uniform hypergraph. Generally, a centrality of an $n$-node $k$-uniform hypergraph $\mathscr{H}(B)$ is defined in the following way:}

\csx{1. Some function $f$ of the centrality of node $u, f\left(c_u\right)$, should be proportional to the weighted sum of some function $g$ of the centrality score of its neighbors, the weights are the same as the weights of an adjacency tensor. In a 3-uniform hypergraph, this means that for some positive constant $\lambda$,
\begin{equation}
    f\left(c_u\right)=\frac{1}{\lambda} \sum_{u,v,w} B_{uvw} g\left(c_v, c_w\right).
\end{equation}}

\csx{2. The centrality scores should be positive, i.e., $c=(c_1,\cdots,c_n)^\top>\mathbf{0}$.}

\csx{Then, we may choose $f\left(c_u\right)=c_u^2$ and $g\left(c_v, c_w\right)=c_v c_w$, which gives $c_u^2=\frac{1}{\lambda} \sum_{u, v, w} B_{uvw} c_v c_w, \Longleftrightarrow B {c}^2= \lambda {c}^{[2]}$. This coincides with the H-eigenproblem \eqref{eq:eigenproblem}. As $c>\mathbf{0}$, the centrality vector $c$ is the Perron-H-eigenvector of an adjacency tensor $B$. The physical meaning of $f\left(c_u\right), g\left(c_v, c_w\right)$ is that $f,g$ should be at the same degree and the contribution $g$ of the centralities of two nodes in a 3-node hyperedge is multiplicative for the third.}
\csx{Generally, let $\mathscr{H}$ be a strongly connected $k$-uniform hypergraph with adjacency tensor $B$. Then the H-eigenvector centrality vector for $\mathscr{H}$ is the positive real vector ${c}$ satisfying 
\begin{equation}\label{eq:centrality}
    B {c}^{k-1}=\lambda {c}^{[k-1]} \text{\quad and \quad} \|{c}\|_1=1
\end{equation}
for some eigenvalue $\lambda>0$. Clearly, $c$ is the normalized ($\|{c}\|_1=1$) Perron-H-eigenvector. Recall that if $A=B-s\mathcal{I}$ is some shift on the $B$, then according to Theorem \ref{thm:perron}, they have the same centrality vector.}

\begin{remark}
    \csx{By leveraging the same idea of defining an eigenvector centrality in \cite{bonacich2007some} for a signed graph, here we may directly define the centrality of a signed uniform hypergraph $\mathscr{H}(A)$ as $A {c}^{k-1}=\lambda {c}^{[k-1]}$ and $\|{c}\|_1=1$, where each entry of $A$ may take a different sign, $\lambda$ is the largest eigenvalue of $A$ and $c$ is the corresponding eigenvector. However, according to Theorem \ref{thm:perron}, the centralities of a Metzler structure defined in both ways (the way of \eqref{eq:centrality} and the way in this remark) are the same. }
\end{remark}

\begin{figure}
        \centering
        \includegraphics[height=3cm]{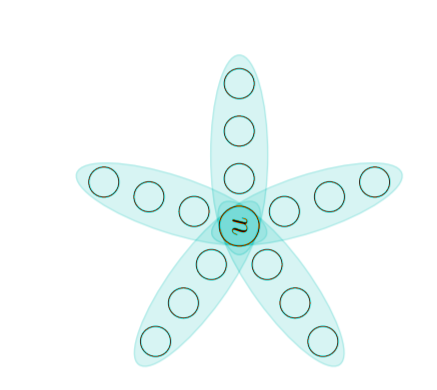}
        \caption{\csx{Example of an 4-uniform, 5-petal(the number of petal) sunflower undirected hypergraph with a singleton core $\{u\}$.}}
        \label{fig:sunflower}
\end{figure}

\csx{In \cite{benson2019three}, a topology named sunflower hypergraph is studied under an unweighted setting. Here, we perform a further analysis under a weighted setting. A sunflower hypergraph has a hyperedge set $E$ with a common pairwise intersection. Formally, for any hyperedges (called petals) $E_i, E_j \in E, E_i \cap E_j=\Tilde{E}$. The common intersection $\Tilde{E}$ is called \emph{the core}. The sunflower hypergraph is similar to the star graph, see figure \ref{fig:sunflower}. In \cite{benson2019three}, the centrality of a $k$-uniform $r$-petal undirected sunflower hypergraph is calculated. Let the central node be $u$ and some other arbitrary node be called $v$. For an unweighted case, $\frac{c_u}{c_v}= r^{\frac{1}{k}}$. Now, let us consider an extreme case of a weighted hypergraph. We set one petal $R$ with an extremely large weight; all other weights remain $0$ or $1$ as in the unweighted case. We can observe that the value of the centrality $c$ will be largely influenced by $R$ since it dominates the H-eigenproblem \eqref{eq:eigenproblem}. In the unweighted case, all the non-central nodes have the same centrality. In the extreme case we consider, the centrality of non-central nodes also has different values, mainly determined by whether the node is in the petal $R$. Thus, we can see that in an unweighted (or equally weighted case), the centrality is mainly determined by the topology of a hypergraph, while in the extremely weighted case (the weights of each edge are very different), the centrality is largely determined by the weights.}

\csx{For a non-uniform hypergraph, there is very limited related literature addressing the concepts of spectrum and centrality. However, a non-uniform hypergraph can be regarded as a multi-layer network, where each layer is a uniform hypergraph \cite{ferraz2021phase}, illustrated in Figure \ref{fig:ill}. We suggest that we utilize the spectrum and centrality of each layer to analyze the whole hypergraph. We will use this idea later to study a system on a non-uniform hypergraph. Alternatively, we can also project a hypergraph into a graph, for graph projection see \cite{ferraz2021phase} for details. It is important to notice that graph projection brings a higher-order hypergraph into a graph (of order two). This will lead to a loss of information. Indeed, by utilizing the graph projection, \cite{ferraz2021phase} is only able to capture the local stability of a system on a hypergraph. Our work will fill this gap and talk about the global one.}

\begin{figure}
        \centering
        \includegraphics[height=5cm]{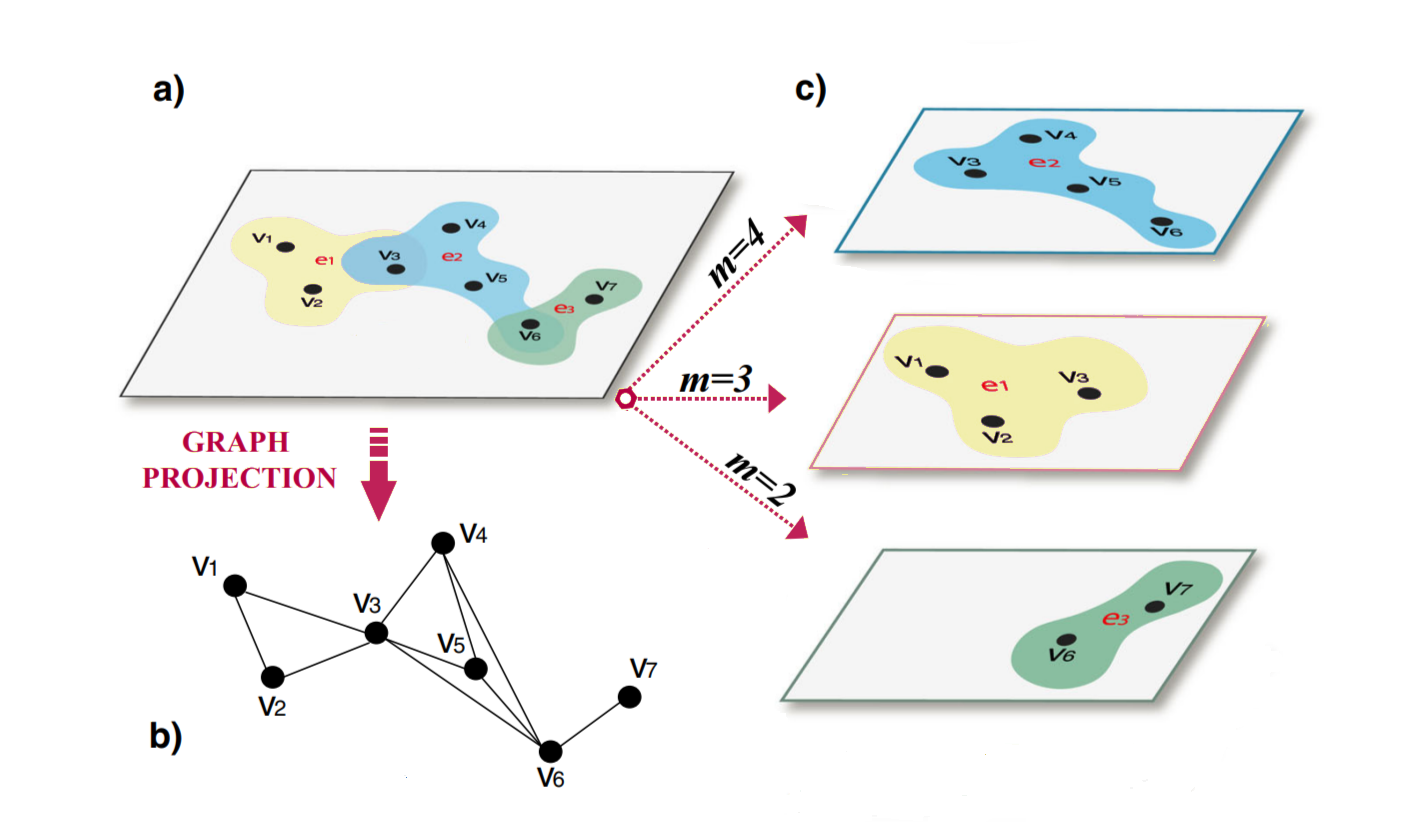}
        \caption{\csx{Illustration figure of a non-uniform undirected hypergraph. The original hypergraph is illustrated in a). The decomposition into different layers with each layer a uniform hypergraph is illustrated in b). The Hypergraph can be projected as a graph as c).}}
        \label{fig:ill}
\end{figure}

\section{Positive Metzler-tensor-based systems on a uniform hypergraph}\label{sec:uni}
Here, we consider a positive Metzler-tensor-based system on a uniform hypergraph of $n$ nodes. The dynamical system is given by
\begin{equation}\label{eq:sys1}
    \dot{x}=Ax^{k-1},
\end{equation}
where $A$ is an irreducible Metzler order $k$ dimension $n$ tensor, and $x\in\mathbb{R}^n$ is the state variable. Component-wise, \eqref{eq:sys1} reads as
\begin{equation}
    \dot{x}_i=\sum_{i_2, \ldots, i_k=1}^n A_{i, I} x_{i_2} \cdots x_{i_k}.
\end{equation}

Firstly, we show that \eqref{eq:sys1} is a positive system.
\begin{lemma}
     The positive orthant $\mathbb R^n_+$ is positively invariant with respect to the flow of \eqref{eq:sys1}.
\end{lemma}

\begin{proof}
    If $x_i=0$, then $\dot{x}_i=\sum_{i_2, \ldots, i_m\neq i} A_{i, I} x_{i_2} \cdots x_{i_k}\geq 0,$ $\forall x_{i_j}\geq0$.
\end{proof}

Now, we are ready to discuss the stability of the origin. Recall that a Metzler tensor can be written as $A=B-s\mathcal{I}$, with $B$ a nonnegative tensor.
\begin{thm}\label{thm:main1}
    The origin is always an equilibrium of \eqref{eq:sys1}. Moreover, in $\mathbb R^n_+$, the origin is: a) { a unique equilibrium and} globally asymptotically stable if  $\lambda(A)=\rho(B)-s<0$; 
    or b) { a unique equilibrium and} unstable if $\lambda(A)=\rho(B)-s> 0$. {Moreover, for the stable case when $k>2$, the bound of convergence rate increases as $|\lambda(A)|$ increases. For the unstable case when $k>2$, the solution diverges to infinity in finite time.}
\end{thm}

\begin{proof}
    The first claim is straightforward.
    Next, we show that the origin is globally asymptotically stable if  $\lambda(A)=\rho(B)-s<0$. Let $\delta=(\delta_1,\cdots,\delta_n)^\top$ be the Perron-eigenvector corresponding to $\lambda(A)$. Define a function as $V_m=\max_i \left(\frac{x_i}{\delta_i}\right)^{k-1}$. For simplicity, we define $m=\argmax_i\left(\frac{x_i}{\delta_i}\right)^{k-1}$. 

    \hma{Since \eqref{eq:sys1} is a positive system, and the Perron-eigenvector is strictly positive,} it holds that $V_m>0$ for any $x\neq 0$, and $V_m=0$ if and only if $x_m=0$. One can see that $V_m$ is continuous, radially unbounded but not continuously differentiable.

    Furthermore, for any $i$, we have
    \begin{equation}
        x_i\leq \max_i \left(\frac{x_i}{\delta_i}\right)^{\frac{k-1}{k-1}}\delta_i=V_m^{\frac{1}{k-1}}\delta_i.
    \end{equation}

Let $T_i=(k-1) x_i^{k-2}$ for any $i$, then we get 
{\small\begin{equation}
    \begin{split}
        \dot{V}_m &= \frac{T_m\dot{x}_m}{\delta_m^{k-1}}=\frac{T_m}{\delta_m^{k-1}}\left(  \sum_{i_2, \ldots, i_k=1}^n A_{m, I} x_{i_2} \cdots x_{i_k} \right)\\
        &= \frac{T_m}{\delta_m^{k-1}}\left(  \sum_{i_2, \ldots, i_k\neq i} A_{m, I} x_{i_2} \cdots x_{i_k} + A_{m,m,\cdots,m}\quad x_m^{k-1}\right)\\
        & \leq T_m\left( \sum_{i_2, \ldots, i_k\neq i} A_{m, I} V_m^{\frac{1}{k-1}} \delta_{i_2} \cdots V_m^{\frac{1}{k-1}} \delta_{i_k} \frac{1}{\delta_m^{k-1}} \right.\\  
        &+  A_{m,m,\cdots,m}\quad V_m \left.\right)\\
        &= T_mV_m \left(A_{m, I} \frac{\delta_{i_2}\cdots \delta_{i_k}}{\delta_i^{k-1}}+ A_{m,\cdots ,m}\right)\\
        &= T_mV_m  \lambda(A)
    \end{split}
\end{equation}}

At the time $x_i=x_m$ and $x_j=x_m$ swap, $\dot{V}_m$ may not exist. However, its Dini-derivatives \eqref{eq:cond} exist and are also non-positive. In light of Lemmas \ref{lem:i1} and \ref{lem:i2} in the appendix, we do not need to deal with such deficiency here.
Therefore, if $\lambda(A)<0$, then $\dot{V}_m\leq0$. If moreover $x=\mathbf{0}$, it holds that $\dot{V}=0$. We further have that $\dot{V}_m\leq T_mV  \lambda(A) \leq 0$. Then, we can obtain that $T_mV_m\lambda(A)=0$ if and only if $x_m=0$. Since $m=\argmax_i\left(\frac{x_i}{\delta_i}\right)^{k-1}$, $x_m=0$ implies $x=\mathbf{0}$. Thus, the set $\{x\,|\,\dot{V}_m(x)=0\}$ only contains the origin, and furthermore the origin is globally asymptotically stable. {Suppose that there is another equilibrium $x^*$. When $x=x^*$, it holds that 
$\dot{x}_m(x)=0$ and thus $\dot{V}_m(x)=0$. This contradicts the fact that the set $\{x\,|\,\dot{V}_m(x)=0\}$ only contains the origin. Thus, the origin is a unique equilibrium.}

{Next, for $k>2$, we can see that the convergence rate to the origin is governed by $\dot{V}_m=\frac{T_m\dot{x}_m}{\delta_m^{k-1}} \leq  T_m V_m  \lambda(A)$. The formula is equivalent to $\dot{x}_m \leq \lambda(A) x_m^{k-1}$, which we can integrate. This leads to $x_m(t)=\left( \frac{1}{2-k}(\lambda(A)t+C_0)\right)^{\frac{1}{2-k}}$, where $C_0= \frac{1}{2-k} x_m^{ \frac{1}{2-k}}(0)$. From this solution, it follows that $\lim_{t\to\infty}x_m(t)=0$ and that the convergence rate increases as $|\lambda(A)|$ increases.
}

Finally, we show the instability of the origin when $\lambda(A)>0$ by Chetaev instability theorem \cite{9778188}.
It is clear that $Ax^{k-1}$ is a polynomial function and thus must satisfy the local Lipschitz condition. Now define a function $\nu_m$ as $\nu_m=\min_i \left( \frac{x_i}{\delta_i}\right)^{k-1}$. We see that $\nu>0$ for any $x> 0$. For simplicity, we now define $r=\arg \min_i\left(\frac{x_i}{\delta_i}\right)^{k-1}$ 

Furthermore, for any $i$, we have
    \begin{equation}
        x_i\geq \min_i \left(\frac{x_i}{\delta_i}\right)^{\frac{k-1}{k-1}}\delta_i=\nu_m^{\frac{1}{k-1}}\delta_i.
    \end{equation}

Then, we get 
\begin{equation}
    \begin{split}
        \dot{\nu}_m &= \frac{T_r\dot{x}_r}{\delta_r^{k-1}}=\frac{T_r}{\delta_r^{k-1}}\left(  \sum_{i_2, \ldots, i_k=1}^n A_{r, I} x_{i_2} \cdots x_{i_k} \right)\\
        & \geq T_r\left( \sum_{i_2, \ldots, i_k\neq m} A_{r, I} \nu_m^{\frac{1}{k-1}} \delta_{i_2} \cdots \nu_m^{\frac{1}{k-1}} \delta_{i_k} \frac{1}{\delta_r^{k-1}} \right.\\  
        &+  A_{r,i,\cdots,i}\quad \nu_m \left.\right)\\
        &= T_r\nu_m \left(A_{r, I} \frac{\delta_{i_2}\cdots \delta_{i_k}}{\delta_r^{k-1}}+ A_{r,\cdots ,r}\right)\\
        &= T_r\nu_m  \lambda(A)\geq 0
    \end{split}
\end{equation}
Moreover, $\dot{\nu}_m=0$ if and only if $x=\mathbf{0}$ by the same argument with the case $\lambda(A)<0$ and $\dot{V}_m\leq0$. By using the Chetaev instability theorem (Theorem 2.1 in \cite{9778188}), one obtains the instability result. {Similar to the stable case, the origin is a unique equilibrium and the solution will diverge to infinity with the divergence rate governed by the best case of $\dot{x}_m = \lambda(A) x_m^{k-1}$. This leads to $x_m(t)=\left( \frac{1}{2-k}(\lambda(A)t+C_0)\right)^{\frac{1}{2-k}}$, where $C_0= \frac{1}{2-k} x_m^{ \frac{1}{2-k}}(0)$ and $k>2$. Since $\lambda(A)>0$, the solution of $x_m$ diverges to infinity within a finite time $t=\frac{-C_0(2-k)}{\lambda(A)}$.}
\end{proof}

\begin{remark}
    The analytical results of Theorem \ref{thm:main1} are similar to those in \cite{chen2022explicit}. However, we emphasize that those results {in \cite{chen2022explicit}} apply to homogeneous polynomial systems on an undirected uniform hypergraph and rely on tensor orthogonal decomposition. However, not every tensor has such decomposition, {for details see \cite{chen2022explicit}.} By contrast, our results apply to any Metzler-tensor-based system (a class of homogeneous polynomial systems) on a directed uniform hypergraph. The stability criterion just requires calculating the Perron-eigenvalue, which is more practical for further application. The Lyapunov-like function $V_m=\max_i \left(\frac{x_i}{\delta_i}\right)^{k-1}$ is very similar to the Lyapunov function used in \cite{rantzer2011distributed,8638525}. In fact, if the order $k=2$ for the matrix case, the Lyapunov function used in \cite{rantzer2011distributed,8638525} is directly recovered. {Concluding, Theorem \ref{thm:main1} shows that the spectrum of the adjacency tensor of a uniform hypergraph plays a decisive role in the stability of a dynamical system based on it. Moreover, the convergence and the divergence rate are also closely related to the spectrum of the adjacency tensor of a uniform hypergraph.} \csx{From a network perspective, the convergence and the divergence rate is closely related to the spectrum of the uniform hypergraph. Moreover, the term $\frac{x_i}{\delta_i}$  denotes a normalized state variable by considering the centrality measure of the node.}
\end{remark}

\begin{exa}\label{exa:1}
    Let the tensor $A$ be a cubical tensor of order $4$ and dimension $4$ with all off-diagonal entries equal to $1$ and all diagonal entries equal to $-64$. One can confirm that $A$ is an irreducible Metzler tensor with Perron-H-eigenvalue $-1$ and Perron-H-eigenvector $\mathbf{1}$. Now consider \eqref{eq:sys1} with this setup. From figure \ref{fig:fig1}, we see that the solution of  \eqref{eq:sys1} indeed converges to the origin from an arbitrary initial condition.
    Then, we let the tensor $\Tilde{A}$ be the tensor 
 of order $4$ and dimension $4$ with all off-diagonal entries equal to $1$ and all diagonal entries equal to $-62$. One can confirm that $\Tilde{A}$ is an irreducible Metzler tensor with Perron-H-eigenvalue $1$ and Perron-H-eigenvector $\mathbf{1}$. A similar simulation verifies that the solution of the system $\Dot{x}=\Tilde{A}x^3$ diverges to infinity. {The simulation is shown in figure \ref{fig:figd}.}
    \begin{figure}
        \centering
        \includegraphics[height=4cm]{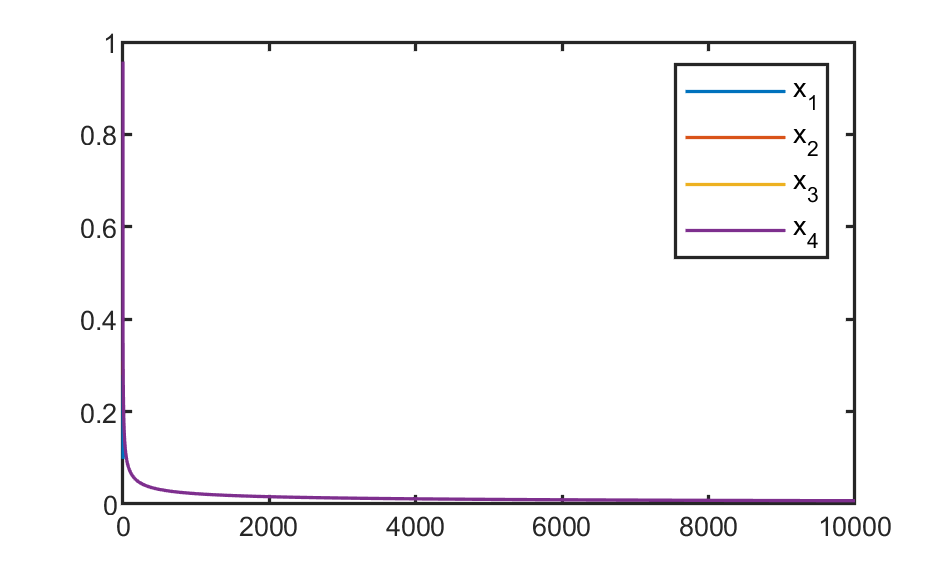}
        \caption{Simulation of system \eqref{eq:sys1} with $\lambda(A)<0$.}
        \label{fig:fig1}
    \end{figure}

    \begin{figure}
        \centering
        \includegraphics[height=4cm]{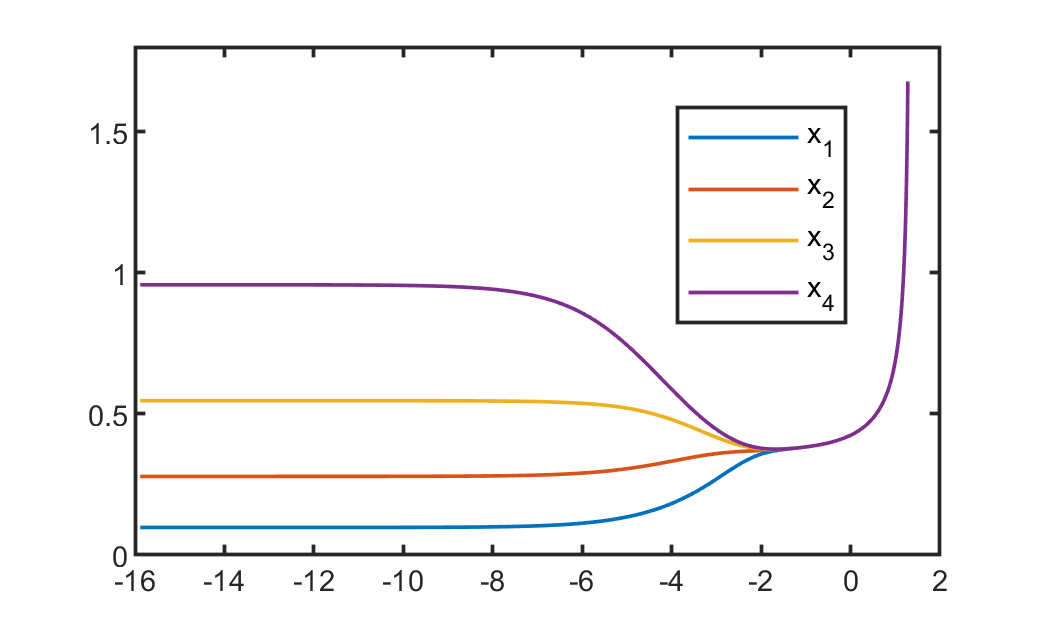}
        \caption{{Simulation of system \eqref{eq:sys1} with $\lambda(A)>0$. The time is presented on a logarithmic scale. The divergence rate is faster than exponential.}}
        \label{fig:figd}
    \end{figure}
\end{exa}

\section{Feedback control strategies}\label{sec:controller}
In this section, we develop feedback control strategies for a positive Metzler-tensor-based system on a uniform hypergraph. The general control goal is to stabilize the origin. Consider system \eqref{eq:sys1} with some control inputs, namely
\begin{equation}\label{eq:control}
\dot{x}=Ax^{k-1}+u. 
\end{equation}

Firstly, we design a feedback control as $u_i={q} x_i^{k-1}$ ($u={q} x^{[k-1]}$). Then, the closed loop system reads as $\dot{x}=Ax^{k-1}+{q}x^{[k-1]}=(A+{q}\mathcal{I})x^{k-1}$. According to Theorems \ref{thm:perron} and \ref{thm:main1}, $\lambda(A+{q}\mathcal{I})=\lambda(A)+{q}$ and if $\lambda(A)+{q}<0$, then we drive the solution of  \eqref{eq:control} asymptotically to the origin. From a perspective of distributed control of a networked system, this control law only requires self-information, which is truly an advantage. \csx{Furthermore, the controller does not change the centrality measure of the nodes. 
} However, the cost to stabilize the system, $q$ (the scalar one needs to manipulate), may be large.

\begin{remark}
    Example \ref{exa:1} already shows the feasibility of this control strategy. Let $A$ and $\Tilde{A}$ be the tensors of Example \ref{exa:1}. The solution of $\Dot{x}=\Tilde{A}x^3$ diverges to infinity. Then, with $u_i=-2x_i^3$, the closed-loop system becomes $\Dot{x}=\Tilde{A}x^3+u=Ax^3$. Thus, the origin is stabilized.
\end{remark}

Next, let us consider a more general feedback controller $u={D}x^{k-1}$ with ${D}$ a non-positive tensor. Then, the closed-loop system reads as $\dot{x}=Ax^{k-1}+{D}x^{k-1}=(A+{D})x^{k-1}$. We know that a Metzler tensor can be written as $A={B}+s\mathcal{I}$, where ${B}$ is a nonnegative tensor and we assume $s<0$ such that $A$ may have negative eigenvalue. Therefore, $A+{D}={B}+{D}+s\mathcal{I}$. Now, suppose $|{D}|\leq {B}$ such that ${D}+{B}$ is still a non-negative tensor and ${D}+{B}\leq {B}$. Thus, we have $\lambda(A+{D})=\rho({D}+{B})+s\mathcal{I}\leq \lambda(A)=\rho(A)+s\mathcal{I}$ according to Lemma \ref{lem:2}. We can see that the Perron-eigenvalue of the closed loop is actually smaller than that of the open loop. Since ${D}+{B}$ is still a non-negative tensor, $\lambda({D}+{B})\geq 0$. Then, according to Lemma \ref{lem:2}, if we increase the entries of $|{D}|$, then $\rho({D}+{B})$ is decreasing but greater than zero. Thus, with a large enough $|{D}|$, we can have $\lambda({D}+{B})+s<0$ and thus $\lambda(A+{D})<0$, indicating that the origin of the closed loop \eqref{eq:control} is globally asymptotically stable. \csx{Notice, in contrast to the previous case, that using the controller potentially changes the centrality measure of the nodes.}

\begin{exa}
    Again, we consider the tensors of example \ref{exa:1} and the system $\Dot{x}=\Tilde{A}x^3$. Now, let the control law be $u=Dx^3$, where $D$ is a tensor of order $4$ and dimension $4$. All off-diagonal entries of the tensor $D$ are $-\frac{1}{2}$ and its diagonal entries are $0$. The closed-loop system is thus $\dot{x}=(\Tilde{A}+D)x^3$. One can see that $\Tilde{A}+D$ is also an irreducible Metzler tensor with Perron-H-eigenvalue $-30.5$ and Perron-H-eigenvector $\mathbf{1}$. Thus, the origin of the closed-loop system $\dot{x}=(\Tilde{A}+D)x^3$ is asymptotically stable.
\end{exa}

\begin{remark}
{We emphasize that the structure ${D}$ is rather arbitrary. This indicates that the proposed control law can be chosen either based on local information or on global information. }
    The term $A+{D}$ can be seen as a modification of the tensor, which can be further regarded as a manipulation of the hypergraph. For example, $A_{ijk}+{D}_{ijk}<A_{ijk}$, may be interpreted as measures taken to reduce the joint influence of $j,k$ to $i$ and it further corresponds to reducing the weight of the hyperedge. In the extreme case, $A_{ijk}+{D}_{ijk}=0$, which means that this hyperedge is removed from the hypergraph. Furthermore, deleting all hyperedges involving an agent is equivalent to removing the corresponding node in the hypergraph. {Thus, we can manipulate the spectrum of the hypergraph by deleting the node or hyperedge and changing some weights of the edges. Then, we are able to further stabilize the system.}
\end{remark}

Finally, we consider the feedback controller $u={D}x^{p-1}$, and $p\neq k$. This leads to the closed loop system $\dot{x}=Ax^{k-1}+{D}x^{p-1}$. The system is now based on a general hypergraph consisting of one layer of a $k$-th order uniform hypergraph and another layer of a $p$-th order uniform hypergraph. Since the closed-loop system is no longer a uniform hypergraph, our analytical results may no longer be applicable but the general problem remains an interesting open problem for future work. The main challenge is how to deal with the combination of tensors with different orders and usually, these two tensors could have different eigenvalues and eigenvectors. However, in the following, we discuss some special cases when the problem can be transformed into a problem on a uniform hypergraph.

\section{Constant control input}

In this section, we consider system \eqref{eq:control} with a constant control input $u=b$ with $b$ a positive vector. This leads to the affine system
\begin{equation}\label{eq:affine}
    \dot{x}=Ax^{k-1}+b.
\end{equation}
This is the simplest in-homogeneous polynomial case. It is easy to check that \eqref{eq:affine} is still a positive system since $b$ is a positive vector. If $b$ is no longer a positive vector, the system may not be a positive system.
Recall that if a Metzler tensor $A$ of order $k$ satisfies $\lambda(A)< 0$, then $-A$ is a non-singular $\mathcal{M}$-tensor.

\begin{lemma}[Theorem 3.2\cite{ding2016solving}]\label{lem:sol}
    If $A$ is a nonsingular $\mathcal{M}$-tensor of order $m$, then for every positive vector $b$ the multilinear system of equations $Ax^{m-1}=b$ has a unique positive solution.
\end{lemma}

Let $\dot{x}=0$ in \eqref{eq:affine}, then we get that $-Ax^{k-1}=b$. Thus, Lemma \ref{lem:sol} directly indicates that the system \eqref{eq:affine} has a unique positive equilibrium if $-A$ is a non-singular $\mathcal{M}$-tensor. 

\begin{lemma}
    Consider system \eqref{eq:affine} with $-A$ being an irreducible non-singular $\mathcal{M}$-tensor. The set $\{x|x\geq x^*\}$, where $x^*$ is the unique positive equilibrium, is a positively invariant set.
\end{lemma}

\begin{proof}
    Firstly, we derive an error dynamic by a change of coordinate $y=x-x^*$: $\dot{y}=A(y+x^*)^{k-1}+b$. Component-wise, we have $\dot{y}_i=\sum_{i_2, \ldots, i_k=1}^n A_{i, I} (y_{i_2}+x_{i_2}^*) \cdots (y_{i_k}+x_{i_k}^*)+b_i$. Now, the statement is equivalent to showing that the error dynamic is a positive system.

    
    Clearly, if $y_i=0$ for an arbitrary $i$, $\dot{y_i}=\sum_{i_2, \ldots, i_k=1} A_{i, I} (y_{i_2}+x_{i_2}^*) \cdots (y_{i_k}+x_{i_k}^*)+b_i$. Since $-A$ is an $\mathcal{M}$-tensor, the only negative term corresponds to the diagonal entry $A_{ii\cdots i}$. Consider the only negative term is $A_{ii\cdots i}(x_i^*)^{k-1}$. We further have that $\sum_{i_2, \ldots, i_k=1}^n A_{i, I} x^*_{i_2} \cdots x^*_{i_k}+b_i=0$ and the only negative term is compensated by some non-negative terms. All the rest of the terms are still non-negative. Thus, $\dot{y_i}\geq 0$, and the error dynamic is a positive system. 
\end{proof}

Now, we investigate the stability of the positive equilibrium.

\begin{thm}\label{thm:main2}
    Consider \eqref{eq:affine} with $-A$ being an irreducible non-singular $\mathcal{M}$-tensor. The unique positive equilibrium $x^*$ is asymptotically stable with a domain of attraction $\{x|x\geq x^*\}$. Furthermore, From almost all initial conditions in $\{x|\mathbf{0}<x< x^*\}$, the solution converges to the unique positive equilibrium $x^*$. 
\end{thm}

\begin{proof}
First of all, we show that $\{x|\mathbf{0}<x< x^*\}$ is a positively invariant set of \eqref{eq:affine}. Since \eqref{eq:affine} is component-wise given by
    $\dot{x}_i=\sum_{i_2, \ldots, i_k=1}^n A_{i, I} x_{i_2} \cdots x_{i_k}+b_i.$

Then, for any $j$,
    $\frac{\partial\dot{x}_i}{\partial x_j}=\sum_{i_2=j} A_{i, I} x_{i_3} \cdots x_{i_k}+\sum_{i_3=j} A_{i, I} x_{i_2}x_{i_4} \cdots x_{i_k}
    +\cdots+\sum_{i_k=j} A_{i, I} x_{i_3} \cdots x_{i_{k-1}}.$
This shows that the Jacobian of the system \eqref{eq:affine} is an irreducible Metzler matrix because all off-diagonal entries are non-negative and the irreducibility of $A$ ensures the irreducibility of the Jacobian. Thus, \eqref{eq:affine} is an irreducible monotone system. From the definition of irreducible monotone systems (section 2.3 \cite{ye2022convergence}), we have $\phi_t(x)< \phi_t(x^*)$ if $x< x^*$. Furthermore, if $x_i=0$, then $\Dot{x}_i\geq b_i>0$. Combining both cases, we confirm that $\{x|\mathbf{0}<x< x^*\}$ is a positively invariant set of the system \eqref{eq:affine}.

Then, we see that in $\{x|\mathbf{0}<x< x^*\}$, there are only one equilibrium, the $x^*$ on the boundary. According to Lemma 2.3 of \cite{ye2022convergence}, the solution converges to the unique positive equilibrium $x^*$ from almost all initial conditions in $\{x|\mathbf{0}<x< x^*\}$.

\begin{figure}
        \centering
        \includegraphics[height=4cm]{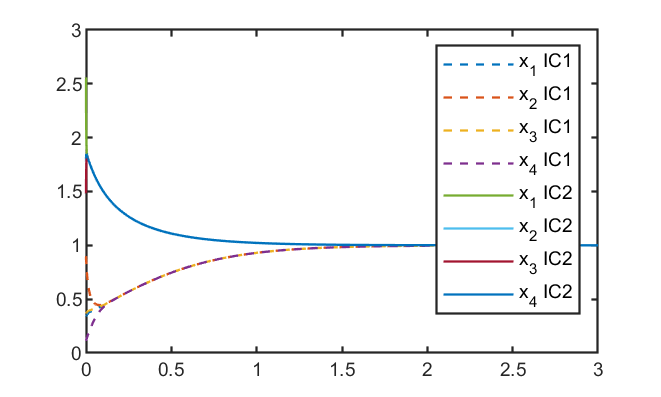}
        \caption{Simulation of system \eqref{eq:affine} for initial conditions in $\{x|\mathbf{0}<x< x^*\}$ (IC1) and in $\{x|x\geq x^*\}$ (IC2).}
        \label{fig:fig2}
    \end{figure}

Finally, we show that the positive equilibrium has another domain of attraction. We define $V_m=\max_i \left(\frac{x_i}{x^*_i}\right)^{k-1}$. Let $m=\arg \max_i\left(\frac{x_i}{\delta_i}\right)^{k-1}$. Let $T_i=(k-1) x_i^{k-2}$, then we get 
{\small\begin{equation}
    \begin{split}
        \dot{V}_m &= \frac{T_m\dot{x}_m}{(x^*_m)^{k-1}}=\frac{T_m}{(x^*_m)^{k-1}}\left(  \sum_{i_2, \ldots, i_k=1}^n A_{m, I} x^*_{i_2} \cdots x^*_{i_k}+b \right)\\
        & \leq \frac{T_m}{(x^*_m)^{k-1}} \left( \sum_{i_2, \ldots, i_k} A_{m, I} V_m^{\frac{1}{k-1}} x^*_{i_2} \cdots V_m^{\frac{1}{k-1}} x^*_{i_k} +b  
        \right)\\
        &= \frac{T_m}{(x^*_m)^{k-1}}(-V_m+1)b
    \end{split}
\end{equation}}
Consider the positively invariant set $\{x|x\geq x^*\}$, which implies $V_m\geq 1$. This further implies $-V_m+1\leq 0$. Since $V_m=0$ implies $x=x^*$, it follows that the domain of attraction is the set $\{x|x\geq x^*\}$.
\end{proof}

\begin{exa}\label{exa:2}
     Let the tensor $A$ be the tensor of Example \ref{exa:1}. Now consider \eqref{eq:affine} with $b=\mathbf{1}$. From figure \ref{fig:fig2}, we see that the solution of \eqref{eq:affine} converges to its unique positive equilibrium from a random initial condition in $\{x|\mathbf{0}<x< x^*\}$. We also see that the solution of the system \eqref{eq:affine} converges to the positive equilibrium from a random initial condition in $\{x|x\geq x^*\}$.
\end{exa}

\section{On a class of Metzler-tensor-based systems on non-uniform hypergraphs}

Usually, systems on a non-uniform hypergraph incorporate tensors of different orders. At the moment, it is mathematically challenging to deal with such systems in full generality. In this section, we investigate a special class of Metzler-tensor-based systems on a non-uniform hypergraph. We show that the techniques developed above are still applicable to this special case. The dynamical system we consider in this section is given by
\begin{equation}\label{eq:sysnon}
    \dot{x}=A_{k-1} x^{k-1}+A_{k-2} x^{k-2}+\cdots +A_{1} x,
\end{equation}
where all the tensor $A_{k-1},\cdots,A_1$ are irreducible Metzler tensors. 

\begin{lemma}
    The positive orthant $\mathbb R^n_+$ is positively invariant with respect to the flow of \eqref{eq:sysnon}.
\end{lemma}

\begin{proof}
    If $x_i=0$, then $\dot{x}_i=\sum_{i_2, \ldots, i_m\neq i} A_{i, I} x_{i_2} \cdots x_{i_k}+ \cdots \geq 0.$
\end{proof}

\begin{thm}\label{thm:nonuni}
    The origin is always an equilibrium of  \eqref{eq:sysnon}. Moreover, in $\mathbb R^n_+$, the origin is {a unique equilibrium and} globally asymptotically stable if  $\lambda(A_i)<0$ for all $i$ and all $\lambda(A_i)$ are associated with the same eigenvector $\delta$. 
\end{thm}

\begin{proof}
    Define the Lyapunov function as $V_m=\max_i \left(\frac{x_i}{\delta_i}\right)^{k-1}$. Let $m=\arg \max_i\left(\frac{x_i}{\delta_i}\right)^{k-1}$. 
    Let $T_i=(k-1) x_i^{k-2}$, then we get 
\begin{equation}
    \begin{split}
        \dot{V}_m &= \frac{T_m\dot{x}_m}{\delta_m^{k-1}}=\frac{T_m}{\delta_m^{k-1}}\left(  \sum_{i_2, \ldots, i_k=1}^n (A_{k-1})_{ m, I} x_{i_2} \cdots x_{i_k} \right.\\
        &+\left. \sum_{i_2, \ldots, i_k=1}^n (A_{k-2})_{ m, i_2 \cdots i_{k-1}} x_{i_2} \cdots x_{i_{k-1}} +\cdots\right)\\
        &\leq \frac{T_m}{\delta_m^{k-1}} \left(  \sum_{i_2, \ldots, i_k=1}^n (A_{k-1})_{ m, I} \delta_{i_2} \cdots \delta_{i_k} V_m \right.\\
        &+\left. \sum_{i_2, \ldots, i_k=1}^n (A_{k-2})_{ m, i_2 \cdots i_{k-1}} \delta_{i_2} \cdots \delta_{i_{k-1}} V_m^{\frac{k-2}{k-1}}+\cdots\right)\\
        &= \frac{T_m}{\delta_m^{k-1}} \left(\sum_{i=1}^{k-1} V_m^{\frac{i}{k-1}} \lambda(A_{i})\delta_m^{i}\right)
    \end{split}
\end{equation}

The rest of the proof remains analogous to the proof of Theorem \ref{thm:main1}.
\end{proof}

\begin{remark}
    \csx{The condition of having the same Perron-eigenvector can be regarded in the following way: since a non-uniform hypergraph consists of several layers and each layer is a uniform hypergraph, then having the same Perron-eigenvector means that each node in each layer of a uniform hypergraph has the same centrality.}
\end{remark}

Recall the following properties of a nonsingular $\mathcal{M}-$tensor.
\begin{lemma}[\cite{ding2016solving}]\label{lem:posxmt}
    If $A$ is a non-singular $\mathcal{M}-$tensor, then there exists a positive vector $x$ such that $Ax^{k-1}>\mathbf{0}$.
\end{lemma}
Recall that if $A$ is an irreducible Metzler tensor and $\lambda(A)<0$, then $-A$ is a nonsigular $\mathcal{M}-$tensor. Thus, there exists a positive vector $x$ such that $Ax^{k-1}<\mathbf{0}$. Thus, we can further refine Theorem \ref{thm:nonuni}.

\begin{thm}\label{thm:nonuni2}
    Consider \eqref{eq:sysnon} {without the restriction that all tensors must be irreducible}. In $\mathbb R^n_+$, the origin is globally asymptotically stable if all $-A_i$ are non-singular $\mathcal{M}-$tensors and there exists one common positive vector $y$ such that $(-A_i) y^i>\mathbf{0}$ for all $i$.
\end{thm}

\begin{proof}
    The proof is similar to the proof of Theorem \ref{thm:nonuni}; one just needs to substitute $\delta$ with the positive vector $y$. 
\end{proof}

\begin{remark}
    Notice that Theorem \ref{thm:nonuni} requires that all $A_i$ have the same Perron-eigenvector, which is restrictive. However, Theorem \ref{thm:nonuni2}, just requires that there exists one common positive vector $y$ such that $(-A_i) y^i>\mathbf{0}$ for all $i$, which is considerably less restrictive. Indeed, according to Lemma \ref{lem:sol}, the number of positive vectors $y$ such that $(-A_i) y^i>\mathbf{0}$ is usually more than the number of Perron-eigenvectors of $A_i$. Thus, Theorem \ref{thm:nonuni2} relaxes the condition of Theorem \ref{thm:nonuni}. {Furthermore, regarding Theorem \ref{thm:nonuni2}, recall that we already assume that there is a common positive vector $y$ such that $(-A_i) y^i>\mathbf{0}$. Therefore, we do not need all tensors to be irreducible. Irreducibility is needed in Theorem \ref{thm:nonuni} to guarantee a positive Perron-eigenvector. In \cite{ding2016solving}, it is shown that if $-A_i$ is a non-singular $\mathcal{M}$-tensor, then $(-A_i) y^i>\mathbf{0}$ for some $y>\mathbf{0}$. Irreducibility is indeed not a need.
    Similarly, we can also derive the statement a) of Theorem \ref{thm:main1} by using this property of a nonsingular $\mathcal{M}$-tensor. This is the special case of Theorem \ref{thm:nonuni2} with one tensor. Finally, we want to emphasize that both Theorems \ref{thm:nonuni} and \ref{thm:nonuni2} have both advantages and disadvantages. Theorem \ref{thm:nonuni} shows that the spectrum of adjacency tensors of a hypergraph plays an important role in the stability of a system on the hypergraph. Tensor's eigenvalues and -vectors can be calculated numerically, which will be introduced later in this paper. The disadvantage is that Theorem \ref{thm:nonuni} requires a stricter condition. In contrast, Theorem \ref{thm:nonuni2} requires a less strict condition but there is no general way to search for a common positive vector $y$. In the later section \ref{sec:cpv}, we provide some results about the common positive vector under some special cases.}
\end{remark}

\begin{remark}
    \csx{The condition of having a common positive vector can be regarded in the following way: if we assume that all the tensors are irreducible, then for each Perron-eigenvector $\delta_i$ of $A_i$: $(-A_i) (\delta_i)^i>\mathbf{0}$ holds and $\delta_i$ denotes the centrality of the layer. For any $x=\delta_i+\Delta$ locally around $\delta_i$, where $\Delta$ is some small perturbation, the $A_i$: $(-A_i) x^i>\mathbf{0}$ should still hold. Thus, in order to have a common $x$, the centrality or the Perron-eigenvector on each layer should not be too different. The more similar the centrality is in each layer, the more chance there is to have a common positive vector required by Theorem \ref{thm:nonuni2}.}
\end{remark}
Furthermore, we show that a class of Metzler-tensor-based systems on a non-uniform hypergraph can be transformed into a Metzler-tensor-based system on a uniform hypergraph. Thus, all the aforementioned results apply also to such a class of Metzler-tensor-based systems on a non-uniform hypergraph. {The feedback control strategies are also applicable to the system after coordinate change.}

Next, we consider the following system:
\begin{equation}\label{eq:shift}
    \dot{x}=A(x-a)^{k-1},
\end{equation}
where $a$ is a positive vector.

\begin{lemma}
    The set $\{x\in\mathbb R^n\, |\, x\geq a\}$ is positively invariant with respect to the flow of \eqref{eq:shift}. 
\end{lemma}

\begin{proof}
    If $x_i=a_i$, then $\dot{x}_i=\sum_{i_2, \ldots, i_m\neq i} A_{i, I} (x_{i_2}-a_{i_2}) \cdots (x_{i_k}-a_{i_k})\geq 0.$ Thus, we complete the proof.
\end{proof}
Component-wise, the system \eqref{eq:shift} is given by,
\begin{equation}
    \dot{x}_i=\sum_{i_2, \ldots, i_k=1}^n A_{i, I} (x_{i_2}-a_{i_2}) \cdots (x_{i_k}-a_{i_k}).
\end{equation}
Then, we define $T=(x_{i_2}-a_{i_2}) \cdots (x_{i_k}-a_{i_k})$ and we have
\begin{equation}
\begin{split}
    T&=x_{i_2} \cdots x_{i_k}+\sum_{p=1;i_2\cdots i_k \neq p}^n  a_{p} x_{i_2} \cdots x_{i_k}\\
    &+\sum_{p,q=1;i_2\cdots i_k \neq p;i_2\cdots i_k \neq q}^n a_{p} a_q x_{i_2} \cdots x_{i_k}\\
    &+\cdots+ a_{i_2} \cdots a_{i_k}.
\end{split}
\end{equation}
Thus, we can rewrite \eqref{eq:shift} as 
{\small\begin{equation}\label{eq:shift2}
    \dot{x}=Ax^{k-1}+B_{k-2}(a,A) x^{k-2}+B_{k-3}(a,A) x^{k-3}+\cdots  +Aa^{k-1}+b,
\end{equation}}\par
\noindent where the entries of $B_{i}$ (for any $i$) are uniquely determined by the vector $a$ and tensor $A$. Each $B_i$ is still a Metzler tensor. The system \eqref{eq:shift2} is based on a non-uniform hypergraph since it contains tensors of different orders. {Conversely, a system in the form of \eqref{eq:shift2} can not be always rewritten as \eqref{eq:shift} and one can always use the method of undetermined coefficients to check the possibility.} By Theorem \ref{thm:main1}, we have the following Corollary.

\begin{cor}\label{cor:main1}
    Consider \eqref{eq:shift}. Then $a$ is always a positive equilibrium of the system \eqref{eq:shift}. Moreover, $a$ is globally asymptotically stable, within the domain of attraction $\{x|x\geq a\}$, if  $\lambda(A)=\rho(B)-s<0$; 
    and unstable if $\lambda(A)=\rho(B)-s> 0$. 
\end{cor}


\begin{proof}
    We consider a change of coordinate:
     \begin{equation}
 \begin{split}
     W: \{x|x\geq a\} \rightarrow \mathbb R^n_+, 
     x \mapsto x-a.
 \end{split}
\end{equation}
By further using Theorem \ref{thm:main1}, we complete the proof.
\end{proof}

\begin{exa}
    Let the tensor $A$ be the tensor of Example \ref{exa:1}. Let the matrix $B$ be a matrix with all off-diagonal entries $1$ and all diagonal entries $-5$. The matrix $B$ is an irreducible Metzler matrix with Perron-eigenvalue $-2$ and Perron-H-eigenvector $\mathbf{1}$.
    Now consider the system $\Dot{x}=Ax^3+Bx$. From figure \ref{fig:fig4}, we see that the solution of the system $\Dot{x}=Ax^3+Bx$ converges to the origin from a random initial condition. 
    \begin{figure}
        \centering
        \includegraphics[height=4cm]{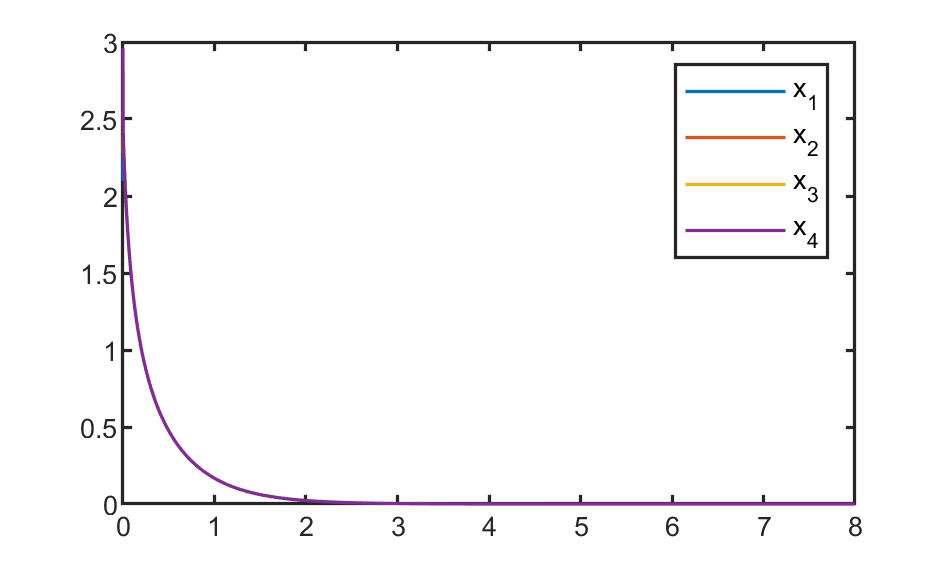}
        \caption{Simulation of the system $\Dot{x}=Ax^3+Bx$ with $\lambda(A)<0$, $\lambda(B)<0$ and $A,B$ having the same Perron-eigenvector.}
        \label{fig:fig4}
    \end{figure}
\end{exa}

\begin{exa}
    \csx{Now, consider a very simple example of Corollary \ref{cor:main1}. Consider a scalar system $\dot{x}=x^2-2x+1$ in the form of \eqref{eq:shift2}. Consider a change of coordinate $z=x-1$; the dynamics become $\dot{z}=z^2$. Thus, the origin of $\dot{z}=z^2$ and the equilibrium $1$ of $\dot{x}=x^2-2x+1$ are unstable. We further consider a feedback controller $g(u)=-2z^2$ as introduced in \ref{sec:controller}. Then, the closed loop system is $\dot{z}=-z^2$. The origin is then stabilized and the $1$ of the original system becomes also asymptotically stable.}
\end{exa}

\section{Application to an SIS epidemic model on a uniform \csx{and a non-uniform} hypergraph}\label{sec:sis}

The properties of Metzler tensors can be further used to study many other dynamical systems on a hypergraph. { Here, we use an SIS model as an example. Although the system is slightly different from the form of a coupled cell system \eqref{eq:cell}, the methods we develop in the paper are still applicable because the system is upper bounded by a coupled cell system \eqref{eq:cell} with Metzler structure. Otherwise, the system can be rewritten as a coupled cell system with some reformulation, though the first approach will be much easier.} In this section, we use an SIS epidemic model on a uniform hypergraph as an illustrative example.

The simplicial SIS model is proposed in \cite{cisneros2021multigroup}, {which can capture the spreading process over a higher-order social network}. It is assumed that $x=(x_1,\ldots,x_n)^\top \in[0,1]^n$. Let $\beta_1, \beta_2\geq 0$, and $\gamma_i>0$ $a_{ij}\geq 0$, ${c}_{ijk}$, $i,j,k \in\{1, \ldots, n\}$. Then, the simplicial SIS model is given by
\begin{equation}\label{eq:sis}
\begin{split}
        \dot{x}_i&=-\gamma_i x_i+\beta_1\left(1-x_i\right) \sum_{j=1}^n a_{i j} x_j\\
&\quad+\beta_2\left(1-x_i\right) \sum_{j, k=1}^n {c}_{i j k} x_j x_k.
\end{split}
\end{equation}
for $i \in\{1, \ldots, n\}$. Now, if we set $\beta_1=0, \beta_2>0$, then \eqref{eq:sis} is based on a uniform hypergraph of 3-body interactions, {which is usually a social network describing both pairwise and group-wise interactions. The social network is a hypergraph constructed through an adjacency matrix $[a_{ij}]$ and an adjacency tensor $[c_{ijk}]$.} Moreover,  \eqref{eq:sis} can be written into the tensor form
\begin{equation}\label{eq:SISt}
    \dot{x}=-Dx+(I-\Dg(x))(\beta_2 {C}x^2),
\end{equation}
where $D=\Dg(\gamma_1,\cdots,\gamma_n)$, ${C}$ is a tensor with entries $b_{ijk}$.

In \cite{cui2023general,cisneros2021multigroup}, only analysis for the case $\beta_1\neq 0$ is provided. In the next, we analyze the case when $\beta_1=0, \beta_2>0$ and consider the following assumption. {Here, we intentionally ignore the pairwise interactions at this moment, which can be regarded as an ideal experiment to see how the spreading process looks like when people can only interact in groups. \csx{Clearly, the real spreading process is a combination of the effect of group-wise interactions and the effect of pairwise interactions. However, one of the key problems of a networked spreading process is to identify the difference between the effect of group-wise interactions and the effect of pairwise interactions. For this purpose, it is meaningful to compare the system behavior under a purely pairwise interaction setting, a purely group-wise interaction setting, and a combination of both.}
Furthermore, since the system behavior for the case $\beta_1> 0, \beta_2=0$ (on a conventional graph), and the case $\beta_1> 0, \beta_2>0$ (on a non-uniform hypergraph) is studied in \cite{cui2023general,cisneros2021multigroup}, combined with the result provided in this section (on a uniform hypergraph), we can see how system's behavior changes with the change of network.}

\begin{assumption}\label{ass:1}
    The matrix $D$ is a positive diagonal matrix, $B$ is an irreducible nonnegative tensor, and $\beta_2>0$.
\end{assumption}

\begin{lemma}
    If Assumption \ref{ass:1} holds, then $[0,1]^n$ is a positively invariant set of the system \eqref{eq:sis}.
\end{lemma}

\begin{proof}
    If $x_i=0$, then $\dot{x}_i=\sum_{i_2, \ldots, i_m\neq i} \beta_2 {c}_{i, I} x_{i_2} \cdots x_{i_k}\geq 0.$ If $x_i=1$, then $\dot{x}_i=-\gamma_i x_i < 0.$
\end{proof}

Next, we show the stability of the healthy state by utilizing the properties of a Metzler tensor. Now, define a tensor $\hat{D}$ and let the diagonal entries $\hat{D}_{iii\cdots}=D_{ii}$ and all off-diagonal entries be zero.


{\begin{cor}\label{cor:sis}
    Consider system \eqref{eq:sis}. If Assumption \ref{ass:1} holds and $\lambda(\beta_2 {C}-\hat{D})<0$, then the healthy state (the origin) is globally asymptotically stable.
\end{cor}}

\begin{proof}
{We notice that $\dot{x}=-Dx+(I-\Dg(x))(\beta_2 {C}x^2) \leq (\beta_2 {C}-\hat{D})x^2$ as $0\leq x_i \leq 1.$ Notice that $\beta_2 {C}-\hat{D}$ is a Metzler tensor with $\lambda(\beta_2 {C}-\hat{D})<0$. We complete the proof from the comparison principle of ODEs and Theorem \ref{thm:main1}. }

\end{proof}

\begin{exa}
    Let the tensor ${C}$ be the tensor of order $3$ and dimension $4$ with all off-diagonal entries $0.01$ and all diagonal entries $0$. Now set $\beta_2=1$ and $D=\Dg[(0.9,0.9,0.9,0.9)^\top]$. One can confirm that $\beta_2{C}-\hat{D}$ is an irreducible Metzler tensor with Perron-H-eigenvalue $-0.75$ and Perron-H-eigenvector $\mathbf{1}$. From Figure \ref{fig:sis}, we see that the solution of \eqref{eq:sis} converges to the origin from a random initial condition.
    \begin{figure}
        \centering
        \includegraphics[height=4cm]{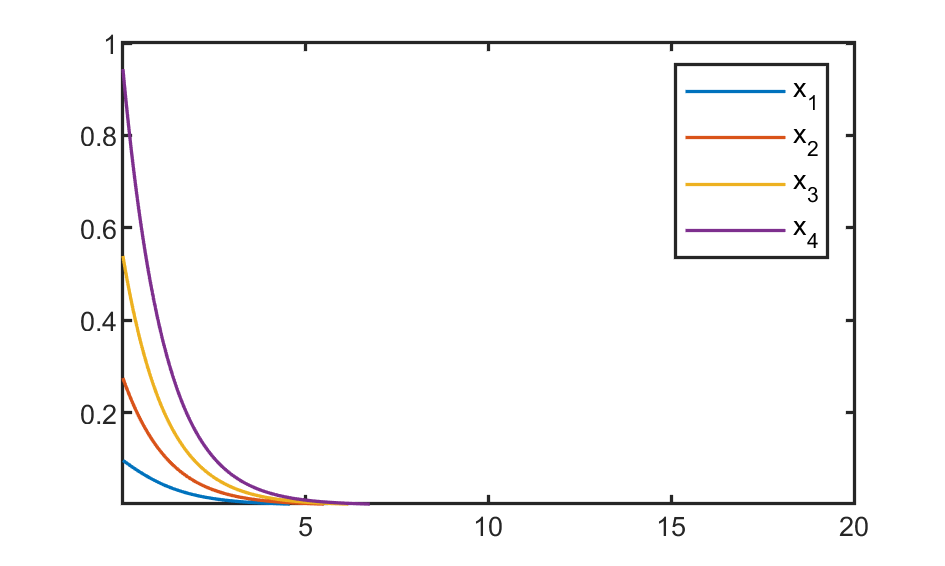}
        \caption{Simulation of an SIS epidemic model on a uniform hypergraph \eqref{eq:SISt}}
        \label{fig:sis}
    \end{figure}
\end{exa}

{Now, we continue to study \eqref{eq:sis} with $\beta_1> 0, \beta_2>0$ (on a non-uniform hypergraph). Some results have already been developed by \cite{cui2023general,cisneros2021multigroup}. The whole dynamics can be written as 
\begin{equation}\label{eq:sisnon}
    \dot{x}=-Dx+(I-\Dg(x))(\beta_1 Ax+\beta_2 {C}x^2),
\end{equation}
where $A$ is a matrix with entries $a_{ij}$.}

\begin{assumption}\label{ass:2}
    {The matrix $D$ is a positive diagonal matrix, $A,{C}$ are nonnegative tensors, and $\beta_1>0, \beta_2>0$.}
\end{assumption}

\begin{lemma}[\cite{cisneros2021multigroup}]
    {If Assumption \ref{ass:2} holds, then $[0,1]^n$ is a positively invariant set of the system \eqref{eq:sisnon}.}
\end{lemma}

\begin{cor}\label{cor:sis2}
    {Consider system \eqref{eq:sis}. If Assumption \ref{ass:1} holds and $\gamma_i> \beta_1 \sum_{j=1}^n a_{i j} + \beta_2 \sum_{j, k=1}^n {c}_{i j k}$ for all $i$, then the healthy state (the origin) is globally asymptotically stable.}
\end{cor}

\begin{proof}
    {As $\gamma_i> \beta_1 \sum_{j=1}^n a_{i j} + \beta_2 \sum_{j, k=1}^n {c}_{i j k}$, there must exist $\hat{\gamma}_i$ and $\Tilde{\gamma}_i$ such that $\hat{\gamma}_i> \beta_1 \sum_{j=1}^n a_{i j}$ and $\Tilde{\gamma}_i> \beta_2 \sum_{j, k=1}^n {c}_{i j k}$. Then, construct the tensors $\hat{D}=\Dg(\hat{\gamma})$ and $\Tilde{D}=\Dg(\Tilde{\gamma})$. 
    We notice that $\dot{x}=-Dx+(I-\Dg(x))(\beta_1 Ax+\beta_2 {C}x^2) \leq (\beta_2 {C}-\Tilde{D})x^2 + (\beta_1 A-\hat{D})x$ as $0\leq x_i \leq 1.$ Notice that the tensors $\beta_2 {C}-\Tilde{D}$ and $\beta_1 A-\hat{D}$ have a common positive vector $\mathbf{1}$, i.e. $(\beta_2 {C}-\Tilde{D})\mathbf{1}^2>\mathbf{0}$ and $(\beta_1 A-\hat{D})\mathbf{1}>\mathbf{0}$. According to Theorem \ref{thm:nonuni2}, one completes the proof.}
\end{proof}

\begin{remark}
    {In \cite[Theorem 5.1]{cisneros2021multigroup},  a sufficient condition for global convergence with respect to the healthy state is also provided, namely \csx{$A$ is irreducible (thus must persist) and $\rho(\beta_1 D^{-1}A+ \beta_2 D^{-1}(\mathbf{1}^\top C_1,\cdots, \mathbf{1}^\top C_n))<1$, where $C_i$ is a matrix with $(C_i)_{jk}=c_{ijk}$.} However, the provided condition is complicated and doesn't have a clear physical meaning. Our result of Corollary \ref{cor:sis2} is different from \cite[Theorem 5.1]{cisneros2021multigroup} since the physical interpretation is straightforward: the origin is globally stable when the healing rate of a node is larger than the sum of all first- and second-order infection rates. \csx{Furthermore, the condition of \cite[Theorem 5.1]{cisneros2021multigroup} cannot be applied to the case when $A$ doesn't persist and our Corollary \ref{cor:sis} fills this gap.}
    Recall that for an SIS on a graph, the condition for a globally stable healthy state is $\lambda(\beta_1 W-D)<0$ where $W,D$ are matrices of infection rates and recovery rates respectively \cite{liu2019analysis}. Combined with Corollary \ref{cor:sis} and Corollary \ref{cor:sis2}, one sees how the healthy domain is changed with the network structure.}
\end{remark}

\section{Application on a cooperative Lotka-Volterra model on a uniform hypergraph}\label{sec:lv}
In this section, we consider a cooperative Lotka-Volterra model on a uniform hypergraph given by
\begin{equation}\label{eq:lv}
    \dot{x}= \Dg(x)(Ax^{k-1}+b),
\end{equation}
where $x\in\mathbb R^n$, $A$ is an irreducible Metzler tensor and $b$ is a positive vector. The model \eqref{eq:lv} is inspired by \cite{singh2021higher} and only takes $k$-body interaction into account. Furthermore, the physical meaning of  $A$ being an irreducible Metzler tensor is that all species cooperate with each other.

It is easy to check the system \eqref{eq:lv} is still a positive system. Since \eqref{eq:lv} is similar to  \eqref{eq:affine}. It follows that \eqref{eq:lv} has a unique positive equilibrium $x^*$, which is the direct consequence of Lemma \ref{lem:sol}. The following proposition confirms its stability.

{\begin{cor}\label{prop:lv}
    Consider system \eqref{eq:affine} with $-A$ a non-singular $\mathcal{M}$-tensor. The unique positive equilibrium $x^*$ is asymptotically stable with a domain of attraction $\{x\,|\,V= \max_i \left(\frac{x_i}{x^*_i}\right)^{k-1}>1\}$. From almost all initial conditions in $\{x\,|\,\mathbf{0}<x< x^*\}$, the solution converges to the unique positive equilibrium $x^*$.
\end{cor}}

\begin{proof}
    The proof is similar to the proof of Theorem \ref{thm:main2}.
\end{proof}

The numerical example of the Lotka-Volterra model is analog to the case of the example \ref{exa:2} and thus omitted {(notice that for a positive system, the flow of $\Dot{x}=\Dg(x)f(x)$ is topologically equivalent to that of $\Dot{x}=f(x)$). The local stability of a generalized Lotka-Volterra model on a non-uniform hypergraph is studied in \cite{cui2023species}. Potentially, tensors are a suitable tool to investigate the global stability of such a system, and it is an interesting future research direction. \csx{Notice that the techniques introduced in this paper will serve as a tool for stability analysis. However, as we do not know whether there is a positive solution in the non-homogenous case (the homogeneous case is guaranteed by Lemma \ref{lem:sol}), the existence of a positive equilibrium must be further investigated.}}

\section{Further discussion}
\subsection{Computation and estimation of H-eigenvalues}\label{sec:heig}

Now, one question remains: how to calculate the H-eigenvalues numerically? Or, more precisely, how to determine whether the Perron-H-eigenvalue is negative?

First of all, if a tensor is supersymmetric, one can check its H-eigenvalues by using the Gershgorin circle theorem for a real supersymmetric tensor (Theorem 6 \cite{qi2005eigenvalues}).

\begin{lemma}[Theorem 6 \cite{qi2005eigenvalues}]\label{lem:circle}
    The eigenvalues of a supersymmetric tensor $A$ lie in the union of $n$ disks in $\mathbb{C}$. These $n$ disks have the diagonal elements of the supersymmetric tensor as their centers, and the sums of the absolute values of the off-diagonal elements as their radii.
\end{lemma}

This leads to the following result.

\begin{lemma}\label{lem:check}
    If $A$ is a supersymmetric and strictly diagonally dominant Metzler tensor with all negative diagonal elements, then all the H-eigenvalues of $A$ are negative.
\end{lemma}

\begin{proof}
Let $\lambda$ be any H-eigenvalue of $A$.
    According to the Gershgorin circle theorem for a supersymmetric tensor (Lemma \ref{lem:circle}) and considering that $A$ is strictly diagonally dominant, we have that, for all $i=1,2, \ldots, n,$
    \begin{equation*}
        \left|\lambda-A_{i i \ldots i}\right| \leq \sum_{\left(i_2, \ldots, i_m\right) \neq(i, \ldots, i)}\left|A_{i i_2 \ldots i_m}\right|<  \left|A_{i i \ldots i}\right|.
    \end{equation*}
    Since all diagonal elements are negative, we obtain $2A_{i i \ldots i}<\lambda<0$.
\end{proof}

Lemma \ref{lem:check} serves as a simple criterion to check whether the Perron-H-eigenvalue is negative. 

In case this simple lemma doesn't apply, we can still use some numerical methods to compute the H-eigenvalues of a tensor, {although this problem is NP-hard}. For example, TenEig is a MATLAB toolbox to find the eigenpairs of a tensor \cite{teneig}. The toolbox is based on the techniques introduced in \cite{chen2016computing}. {However, this technique may not be numerically stable when the order or the dimension of the tensor is extremely large. It remains an open question to develop efficient and stable algorithms.}

\subsection{Common positive vectors}\label{sec:cpv}

{The following lemma further provides a criterion to check when the Theorem \ref{thm:nonuni2} is applicable.}
\begin{lemma}\label{lem:dom}
    {If $A\in\mathbb R^{[k,n]}$ is a strictly diagonally dominant Metzler tensor with all negative diagonal elements, then it holds that $-A\mathbf{1}^{k-1}>\mathbf{0}$.}
\end{lemma}

\begin{proof}
  {  From the definition of diagonal dominance we have that $
\left|A_{i i \ldots i}\right| > \sum_{\left(i_2, \ldots, i_m\right) \neq(i, \ldots, i)}\left|A_{i i_2 \ldots i_m}\right| \quad \text { for all } i=1,2, \ldots, n
$. This readily leads to $-A\mathbf{1}^{k-1}>\mathbf{0}$.}
\end{proof}

{Lemma \ref{lem:dom} shows that if all tensors in Theorem \ref{thm:nonuni2} are strictly diagonally dominant Metzler tensors with all diagonal elements negative, then $\mathbf{1}\in\mathbb R^n$ is a common positive vector required by Theorem \ref{thm:nonuni2}. In this setting, Theorem \ref{thm:nonuni2} is directly applicable and there is no need to calculate the H-eigenvalues. }


\section{Conclusions}
In this paper, we study a series of positive systems on hypergraphs constructed from a Metzler tensor. The concept of a Metzler tensor is a generalization of the concept of a Metzler matrix. The main properties of a Metzler tensor are given by the Perron–Frobenius Theorem \ref{thm:perron}. We further use this Theorem to study the stability of a class of positive systems on both uniform and non-uniform hypergraphs, especially for two typical systems (a higher-order Lotka Volterra system and a higher-order SIS process). Moreover, we develop some feedback control strategies to stabilize the origin and investigate a system with a constant control input. All theoretical results are illustrated by some numerical examples. Finally, we illustrate some criteria and numerical methods to check the H-eigenvalues of a tensor.

There are several questions still open for future work. One is how to generalize the results for systems on a more general non-uniform hypergraph. Another is how to investigate the dynamical system on a uniform hypergraph, where the adjacency tensor is not a Metzler tensor. 

\bibliographystyle{IEEEtran}
\bibliography{bib}

\begin{appendix}
We consider discontinuous dynamical systems governed by differential equations of the form
$
\dot{x}=f(x),
$
where $f: \mathbb{R}^n \rightarrow \mathbb{R}^n$ is a piecewise continuous function that undergoes discontinuities on a set $N$ of measure zero. 

\begin{lemma}[\cite{alvarez2000invariance}]\label{lem:i1}
    Suppose that there exists a positive definite, Lipschitz-continuous function $V(x)$ such that
\begin{equation}\label{eq:cond}
\frac{d}{d t} V(x(t))=\left.\frac{d}{d h} V(x(t)+h \dot{x}(t))\right|_{h=0} \leqslant 0
\end{equation}
almost everywhere. Let $M$ be the largest invariant subset of the manifold where the strict equality \eqref{eq:cond} holds, i.e., the set $\left\{x\in\mathbb R^n\,:\,  \frac{d}{d t} V(x(t))=0\right\}$, and let $V(x) \rightarrow \infty$ as $\operatorname{dist}(x, M) \rightarrow \infty$. Then all the trajectories $x(t)$ of $(1)$ converge to $M$, that is,
$$
\lim _{t \rightarrow \infty} \operatorname{dist}(x(t), M)=0.
$$
\end{lemma}

$ $

\begin{lemma}[\cite{alvarez2000invariance}]\label{lem:i2}
     Condition \eqref{eq:cond} of Theorem 1 is fulfilled if $\frac{d}{d t} V(x(t))=\left.\frac{d}{d h} V(x(t)+h \dot{x}(t))\right|_{h=0}$ is nonpositive at the points of the set $N_V$ where the gradient $\nabla V$ of the function $V(x)$ does not exist, and in the continuity domains of the function $f(x)$ where (4) is expressed in the standard form
$$
\frac{d}{d t} V(x)=\nabla V(x) \cdot f(x), \quad x \in \mathbb{R}^n \backslash\left(N \cup N_V\right).
$$
\end{lemma}

$ $

Although these two Lemma are designed for a discontinuous dynamical system, we can still use them for a continuous dynamical system if we want to adopt a not continuously differentiable Lyapunov function.
\end{appendix}



\end{document}